\documentclass[11pt]{article}
\usepackage{amsmath,amsthm,amsfonts,amssymb}
\usepackage{fullpage}
\usepackage{hyperref}
\usepackage{url}
\usepackage{amscd,amssymb,amsmath,amsthm,latexsym,enumerate}
\usepackage{tikz,tikz-3dplot}
\usetikzlibrary{quantikz}
\usepackage{tikz-network}
\usepackage{graphicx}
\usepackage{array}
\usepackage{mathtools}
\usepackage{cuted}
\usepackage{xcolor}

\newif\ifcomments
\commentsfalse

\newcommand{\comment}[1]{\ifcomments{\color{blue}\texttt{{[}COMMENT:} {#1} \texttt{END-COMMENT{]}}}\fi}

\allowdisplaybreaks

\newcommand{\Mod}{\textup{Mod}}

\newcommand{\SWAP}{\textit{SWAP}}
\makeatletter
\newcommand{\tpmod}[1]{{\@displayfalse\pmod{#1}}}
\makeatother

\theoremstyle{plain}
\newtheorem{Thm}{Theorem}[section]

\theoremstyle{definition}

\title{Quantum Fanout and GHZ states using spin-exchange interactions}

\author{Stephen Fenner\thanks{Computer Science and Engineering Department, Columbia, SC 29208.  \url{fenner@cse.sc.edu}, \url{rwosti@email.sc.edu}} \\ University of South Carolina
\and
Rabins Wosti$\!\,^*$ \\ University of South Carolina}

\begin{document}
%\bstctlcite{IEEEexample:BSTcontrol}
\maketitle

%\begin{abstract}
%It has been shown that, for any $n > 0$, if $2n$ physical qubits (as spin-$1/2$ particles) are suitably encoded into $n$ logical qubits, then the Heisenberg interactions (either isotropic or nonisotropic) among $2n$ qubits can be used to exactly implement an $(n + 1)$-qubit fanout gate using a certain constant depth circuit [arXiv:quant-ph/0407125]. However, the coupling coefficients in the Hamiltonian considered in that paper are assumed to be all equal. In this current work, we generalize these results and show that for all $n > 0$, one can exactly implement an $(n + 1)$-qubit parity gate and hence, equivalently in constant depth an $(n+1)$-qubit fanout gate, using a similar Hamiltonian but with unequal couplings, and we give an exact characterization of which couplings are adequate to implement fanout via the same circuit. More precisely, we show that for any fixed qubit $q$ among the $n$ data qubits and a suitable time $T$ of evolution, the coupling of $q$ with any other qubit should be an odd integer multiple of $\pi/2T$. This work resolves a question left open in the paper [arXiv:quant-ph/0407125].
%\end{abstract}

\begin{abstract}
We show how the fanout operation on $n$ logical qubits can be implemented via spin-exchange (Heisenberg) interactions between $2n$ physical qubits, together with a physical target qubit and $1$- and $2$-qubit gates in constant depth.  We also show that the same interactions can be used to implement $\Mod_q$ gates for any $q>1$.  These results allow for unequal coupling strengths between physical qubits.  This work generalizes an earlier result by Fenner \& Zhang \cite{Fenner_Zhang}, wherein the authors showed similar results assuming all pairwise couplings are equal.  The current results give exact conditions on the pairwise couplings that allow for this implementation. Precisely, each logical qubit is encoded into two physical qubits. Couplings between physical qubits encoding the same logical qubit are termed as internal couplings and couplings between the ones encoding different logical qubits are termed as external couplings.
We show that for a suitable time $T$ of evolution, the following conditions should hold: a) every external coupling should be an odd integer multiple of $\pi/2T$; b) every internal coupling should be an integer multiple of $\pi/T$; and c) the external magnetic strength in $z$-direction should be an integer multiple of $\pi/T$. Since generalized GHZ (``cat'') states can be created in constant depth using fanout, the same interactions can be used to create these states.
\end{abstract}


\section{Introduction}
The accurate computation of advanced quantum algorithms like Shor's integer factorization, quantum phase estimation (QPE), and the quantum Fourier transform (QFT) requires quantum circuits of considerable size and depth.  It is difficult to achieve reliable computation with deep quantum circuits due to the limited coherence times of the current noisy quantum devices.  The quantum fanout gate is known to be a powerful primitive for reducing the depth of many quantum circuits~\cite{Hoyer_Spalek,Gottesman}.  Shallow or constant-depth quantum circuits are desirable for both near-term and fault-tolerant quantum computations as they reduce noise and allow faster execution of quantum algorithms, potentially skirting the effects of short coherence times.  The $n$-qubit quantum fanout gate $F_n$ for $n \geq 1$ is the $(n+1)$-qubit unitary operator that copies the classical bit value of the control qubit into $n$ target qubits as shown below:
$$
F_n\ket{x_1, \ldots, x_n, c} = \ket{x_1 \oplus c, \ldots, x_n \oplus c, c},
$$
for all $x_1, \ldots, x_n, c \in \{0,1\}$. When the control qubit $\ket{c}$ is in the state $\alpha\ket{0} + \beta\ket{1}$, fanout can be used to construct ``cat'' states of the form $\alpha\ket{0}^{\otimes n+1} + \beta\ket{1}^{\otimes n+1}$. When $\alpha = \beta = \frac{1}{\sqrt 2}$, these states are called GHZ states. The $n$-qubit parity gate $P_n$ for $n \geq 1$ is the $(n+1)$-qubit unitary operator such that $P_n\ket{x_1, \ldots, x_n, t} = \ket{x_1, \ldots, x_n, t \oplus x_1 \oplus \cdots \oplus x_n}$ for all $x_1, \ldots, x_n, t \in \{0,1\}$. It was shown by Moore \cite{Moore} that the fanout gate can be converted into a parity gate in constant depth by conjugating with a bank of Hadamard ($H$) gates and vice-versa, that is, $H^{\otimes (n+1)}F_nH^{\otimes (n+1)} = P_n$. The fanout gate can be implemented by a $O(\log n)$-depth circuit with $O(n)$ many CNOTs. However, a constant-depth implementation of fanout for creating long-range entanglement among the qubits can significantly reduce the circuit depth of several quantum algorithms: a) one can approximate QPE, QFT, sorting, and arithmetic operations in constant depth and polynomial size \cite{Hoyer_Spalek}; b) since QFT forms a subroutine in the Shor's factoring algorithm, one can execute the entirety of Shor's algorithm in constant depth; c) one can exactly implement $n$-qubit threshold gates, unbounded AND-gates (generalized Toffoli gates), and OR-gates in constant depth \cite{Takahashi}.  Since GHZ states are widely used as resource states in quantum error correction protocols, e.g., to read out the error syndrome in Shor's code \cite{Shor}, and quantum teleportation \cite{Gottesman}, fanout proves helpful in reducing the associated circuit depth.  A key advantage of fanout is that it allows any set of commuting unitary operators to be simultaneously diagonalized and applied in parallel in constant depth and polynomial size, even if those operators act on the same qubits \cite{Hoyer_Spalek}.  Recently, quantum fanout gates are shown applicable to construct constant-depth and polynomial-size quantum circuits to realize quantum memory devices, including quantum random access memory (QRAM) and quantum random access gates (QRAGs) \cite{Allcock, Rosenthal2021}.

However, there is mounting theoretical evidence that fanout gates cannot be implemented using conventional quantum circuits in small (sublogarithmic or constant) depth and small width, even if unbounded AND-gates are allowed \cite{Fang, Rosenthal2020}. More recently, there have been several constant-depth implementations of the fanout gate \cite{Pham, Burhman, Piroli, Baumer} using $O(n)$ ancilla qubits, however these are based on measurement-based quantum computation and classical feedback, which may be an inaccessible resource in certain near-term experimental systems \cite{Arute}, and also can introduce additional measurement overheads and errors.  These results suggest that perhaps implementing fast quantum fanout gates requires more unconventional approaches.

Recent literature shows that by rightly tuning the interactions among the qubits in a quantum Heisenberg model, one can simulate many-body quantum spin systems, which has led to discoveries of several interesting physical phenomena \cite{Salathe,Savary,Bernien,Jepsen,Nguyen}. The Heisenberg interactions also form the primitives to implement expressive multi-qubit gates \cite{Kivlichan} which are extensively used in quantum algorithms and quantum error correction. So, an alternate approach would be to evolve an $n$-qubit system according to a Hamiltonian that occurs in nature, along with a minimal number of traditional quantum gates. It was shown by Fenner \cite{Fenner} that fanout can be realized by using a simple Ising-like Hamiltonian of equal coupling strengths and was later generalized to unequal couplings by Fenner \& Wosti \cite{Fenner_Rabins}. Guo et al. \cite{Andrew} showed an implementation of an unbounded quantum fanout gate by using a sequence of CNOTs, which are produced by systematically evolving pairwise Hamiltonian terms that employ power-law interactions between the qubits. Fenner \& Zhang \cite{Fenner_Zhang} showed that, for any $n > 0$, if $n$ logical qubits are suitably encoded into $2n$ physical qubits (as spin-$1/2$ particles), then the Heisenberg interactions among $2n$ qubits can be used to exactly implement an $(n + 1)$-qubit fanout gate using a certain constant depth circuit. However, the coupling coefficients in the Heisenberg Hamiltonian considered in \cite{Fenner_Zhang} are assumed to be all equal, which is physically unrealistic given that in real physical systems, we expect the couplings to be stronger between the spins that are in close proximity than the ones spatially far apart. Furthermore, application of multi-qubit gates on distant qubits in the current large scale quantum computing architectures with limited interqubit connectivity  introduces overheads that counteract the speed-up benefits of using long-range entangling gates \cite{Monroe,Linke,Bapat,Childs}. So, using Hamiltonians with long-range unequal interactions can help in realizing these multi-qubit entangling gates in small depth. Some of these longe-range interactions occur naturally in physical systems like dipole dipole and van der Waals interactions between Rydberg atoms \cite{Saffman, Weimer}, dipole-dipole interactions between polar molecules \cite{Yan}, and between defect centers in diamond \cite{Yao}. To this end, in this current work, we generalize the results of \cite{Fenner_Zhang} and show that for all $n > 0$, one can exactly implement an $(n + 1)$-qubit parity gate and hence, equivalently in constant depth an $(n+1)$-qubit fanout gate, using a similar Heisenberg Hamiltonian but with unequal couplings,
\comment{Say that, the current work allows for logical qubits that are spatially separated.  Quantum dots, perhaps?}
and we give an exact characterization of which couplings are adequate to implement fanout via the same circuit. This possibly allows for the logical qubits to be spatially separated in experimental quantum systems, for example, quantum dots to realize fanout. This work resolves a question left open in \cite{Fenner_Zhang}. Additionally, we give a direct implementation of $\Mod_q$ gate for $q \geq 2$ using our modified Hamiltonian by generalizing the constraints on the couplings that we derive to implement the parity gate $(q=2)$.
%In particular, we show that for any fixed qubit $q$ among the $n$ data qubits and a suitable time $T$ of evolution, the coupling of $q$ with any other qubit should be an odd integer multiple of $\pi/2T$. This work resolves a question left open in \cite{Fenner_Zhang}. 

\section{Heisenberg interactions}
The Heisenberg interaction describes the way particles in the same general location affect each other by the magnetic moments arising from their spin angular momenta. We assume that the physical qubits are implemented as spin-$1/2$ particles, with $\ket{0}$ being the spin-up state (in the positive $z$-direction) and $\ket{1}$ being the spin-down state (in the negative $z$-direction). Given a system of $n$ identical qubits (spins) labeled $1, \ldots, m$, we define
\begin{align}
&J_x = \frac{1}{2} \sum_{i=1}^n X_i \nonumber \\
&J_y = \frac{1}{2} \sum_{i=1}^n Y_i \nonumber \\
&J_z = \frac{1}{2} \sum_{i=1}^n Z_i, \label{z-observable}
\end{align}
where $X_i, Y_i,$ and $Z_i$ represent the three Pauli operators acting on the $i^{th}$ qubit. The three observables $J_x, J_y,$ and $J_z$ give the total spin of $n$ qubits in the $x$-, $y$-, and $z$-directions, respectively. Consequently, one can define weighted quadratic versions of $J_x, J_y,$ and $J_z$ as
\begin{align}
&K_x = \frac{1}{2}\sum_{1 \leq i<j \leq n} J_{ij} X_iX_j \nonumber \\
&K_y = \frac{1}{2} \sum_{1 \leq i<j \leq n} J_{ij} Y_iY_j \nonumber \\
&K_z = \frac{1}{2} \sum_{1 \leq i<j \leq n} J_{ij} Z_iZ_j, \nonumber
\end{align}
where each $J_{ij}$ represents the coupling strength between the $i^{th}$ and the $j^{th}$ spins.
A weighted analogue of the squared magnitude of the total spin angular momentum of the system is given by the observable
\begin{align}
J^2 & = K_x + K_y + K_z \nonumber \\
& = \frac{1}{2}\sum_{1 \leq i<j \leq n} J_{ij} (X_iX_j + Y_iY_j + Z_iZ_j) \label{initial_J2}
%& = \frac{3n}{4}I + \sum_{1 \leq i<j \leq n} \vec{S_i} \cdot \vec{S_j}
\end{align}
To account for an external magnetic field in the $z$-direction, for any real $g$, we define the Hamiltonian
\begin{equation}
H_g = -J^2 + gJ_z. \label{Hamiltonian}
\end{equation}
Unencoded computational basis states are generally not eigenstates of $H_g$.  In order to extract the parity of a computational basis state easily, we wish to encode it as an eigenstate of $H_g$.  We also wish to maximize the locality of the encoding by independently acting on groups of physical qubits that are as small as possible.  Furthermore, it is desirable to accomplish this by encoding the basis state as a simultaneous eigenstate of both $J^2$ and $J_z$.

Consider a $p$-bit binary string $x = x_1\cdots x_p$ for some integer $p \geq 2$. Define
\begin{align*}
C_0^x &:= \{ i\in [1,p] : x_i = 0\} \\
C_1^x &:= \{ i\in [1,p] : x_i = 1\},
\end{align*}
and let $w$ denote the Hamming weight of $x$, i.e., $w = |C_1^x|$. 
%Without loss of generality, assume that $x = 1^{w}0^{p-w}$. Let $a$ and $b$ denote the position of the first and the last bit set to 1 in $x$ respectively. Similarly, let $c$ and $d$ denote the position of the first and the last bit set to 0 in $x$ respectively. 
Define $\ket{x'} := \bigotimes_{i=1}^p \ket{x_i0}$, where every bit $x_i$ is paired with an ancilla qubit set to state $\ket{0}$.  For any variable $a \in [1,p]$, $a$ indexes the pair of qubits with $a_1 = 2a-1$ representing the first qubit and $a_2 = 2a$ representing the second qubit in the indexed pair. Consider an encoding unitary $E$, such as that depicted in Figure~\ref{fig:encoder}, that produces the following outputs:
\begin{align}
\ket{\psi_0} &:= E\ket{00} = \ket{00},   \nonumber \\
\ket{\psi_1} &:= E\ket{10} = (\ket{01} - \ket{10})/\sqrt 2.
\end{align}
%(Figure~\ref{fig:encoder}
\begin{figure}
\begin{center}
\begin{quantikz}
&\ctrl{1} & \gate{H}  & \ctrl{1} & \qw \\
&\targ{}  & \ctrl{-1} & \targ{}  & \qw
\end{quantikz}
\end{center}
\caption{A circuit implementing the encoder $E$.}\label{fig:encoder}
\end{figure}
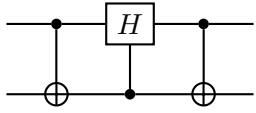
%shows a possible circuit implementing $E$.)
Let $n := 2p$.  The unitary $E^{\otimes p}$ encodes the state $\ket{x'}$ as the $n$-qubit state
\begin{equation} \label{encoded_state}
\ket{x_L} := E^{\otimes p}\ket{x'} = \bigotimes_{i=1}^p  \ket{\psi_{x_i}}. 
\end{equation}
$E$ is chosen to act on as few qubits as possible while preserving the Hamming weight of its input state, and it is evident that $\ket{x_L}$ is an eigenstate of $J_z$ with eigenvalue $p-w$.
\comment{$\ket{x_L}$ is also a spin state, annihilated by the raising operator.  Can we motivate this fact at all?}
The observable $J^2$ given by Eq.~(\ref{initial_J2}) can be re-written as,
\begin{align}
%J^2 = & \frac{3n}{4} I + \frac{1}{2}\sum_{1 \leq i < j \leq n} J_{ij} (X_iX_j + Y_iY_j + Z_iZ_j) \nonumber \\
J^2 = & \eta I + \frac{1}{2}\sum_{1 \leq i < j \leq n} J_{ij} (I_iI_j + X_iX_j + Y_iY_j + Z_iZ_j) \nonumber \\
= & \eta I + \sum_{1 \leq i < j \leq n} J_{ij} \SWAP_{ij} \label{modified_J2}
\end{align}
where $\eta = - \frac{1}{2}\sum_{1 \leq i < j \leq n} J_{ij}$.
%\comment{Do we really need the next assumption?} We assume that each $J_{ij}$ is nonzero.
We are interested in finding the constraints on the couplings $J_{ij}$ so that the encoded state $\ket{x_L}$ is an eigenstate of $J^2$, i.e.,
\begin{equation} \label{eigen_equation}
J^2\ket{x_L} = \lambda_x \ket{x_L}
\end{equation}
for some real eigenvalue $\lambda_x$.  We call the coupling between two qubits of the same pair an \emph{internal coupling}, and the coupling between two qubits in distinct pairs an \emph{external coupling}. Solving Eq.~(\ref{eigen_equation}) for $\lambda_x$ gives two conclusions: 
\begin{enumerate} 
\item For any distinct qubit pairs $u,v \in [1,p]$, all the external couplings are equal, i.e. $J_{u_1v_1} = J_{u_1v_2} = J_{u_2v_1} = J_{u_2v_2}$.
\item The eigenvalue $\lambda_x$ is given in terms of coupling strengths as follows:
\begin{align}
\lambda_x &= 2\left(\sum_{r\in C_1^x} \sum_{t\in C_0^x} J_{r_1t_1} \right) \nonumber\\
&\mbox{} + 2\left(\sum_{\substack{r,s\in C_1^x\\r<s}} J_{r_1s_1} \right) + 4\left(\sum_{\substack{m,n\in C_0^x\\m<n}} J_{m_1n_1}\right) \nonumber \\
&\mbox{} + \left(\sum_{m\in C_0^x} J_{m_1m_2} \right) 
- \left(\sum_{r\in C_1^x} J_{r_1r_2}\right) \label{final_eigenvalue_eqn2}
\end{align}
Refer to Appendix~\ref{lambda_solution} for the detailed calculations.
\end{enumerate}
%Based on (1) above, we hereafter set $\alpha = \frac{1}{\sqrt 2}$ and $\beta = -\frac{1}{\sqrt 2}$, so that $\ket{\psi_1} = (\ket{01}-\ket{10})/\sqrt 2$, the singlet Bell state.

%\comment{Maybe we can move the next theorem to the appendix with a forward reference in the paragraph that follows.  (I just rewrote it.)}
%Refer to Appendix~\ref{Hamming_weight_dependence} for the proof of Theorem~\ref{Thm_Hamming_weight}.

We show that requiring the eigenvalue $\lambda_x$ of $J^2$ to depend only on the Hamming weight of the input string $x$ reduces the couplings among $n$ qubits to only two (possibly unequal) internal and external couplings, which leads to an unrealistic constraint on the Hamiltonian (see Theorem~\ref{Thm_Hamming_weight} and its proof in Appendix~\ref{Hamming_weight_dependence}). In the next section, we relax this requirement to obtain a looser constraint on the couplings.
\section{Relaxing the Hamming weight dependence of eigenvalues}
%From here, we relax the constraint that the eigenvalues depend only on the Hamming weight of the input string. All the calculations below will be based on the sections prior to and up to Eq.~(\ref{final_eigenvalue_eqn}).
%
Consider our Hamiltonian $H_g$ as defined in Eq.~(\ref{Hamiltonian}). 
%It can be shown that $J^2$ (Eq.~(\ref{modified_J2}) and $J_z$ (Eq.~(\ref{z-observable})) commute with each other.
\iffalse
Let $C_1^{x}$ and $C_0^{x}$ denote the sets of qubit pair indices in $\ket{x_L}$ corresponding to the states $\ket{\psi_1}$ and $\ket{\psi_0}$, respectively.
Eq.~(\ref{final_eigenvalue_eqn2}) can now be written in a more general setting as below:
\begin{align}
& \lambda_x =  2\left(\sum_{r \in C_1^{x}} \ \sum_{t \in C_0^{x}} J_{r_1t_1} \right) + 2\left(\sum_{\substack{r,s \in C_1^{x} \\ s>r}} J_{r_1s_1} \right) \nonumber \\
& + 4\left(\sum_{\substack{m,n \in C_0^{x} \\ n>m}} J_{m_1n_1}\right)
+ \left(\sum_{m \in C_0^{x}} J_{m_1m_2} \right) \nonumber \\
&  - \left(\sum_{r \in C_1^{x}} J_{r_1r_2}\right) \label{modified_eigenvalue_eqn}
\end{align}
\fi
Noting that
\begin{align*}
(Z_1 + Z_2) \ket{\psi_0} &= 2 \ket{\psi_0}, \\
(Z_1 + Z_2) \ket{\psi_1} &= 0,
\end{align*}
we have
\begin{equation}
J_z \ket{x_L} = (p-wt(x))\ket{x_L}.
\end{equation}
Therefore,
\begin{align}
H_g \ket{x_L} &= \left(-\lambda_x + g(p-wt(x)) \right) \ket{x_L} \nonumber \\
&= \delta_x \ket{x_L}, \label{eqn:H_g-again}
\end{align}
where
\begin{equation}\label{eqn:delta_x}
\delta_x = -\lambda_x + g(p-wt(x)).
\end{equation}
Let $U = e^{-iTH_{g}}$ for some $T$ of evolution.
Then, $U \ket{x_L} = e^{-iTH_{g}}\ket{x_L} = e^{-iT\delta_x}\ket{x_L}$.
\iffalse
Choose an arbitrary pair index $z \in C_1^{x}$.
Consider another string $y$ such that $y_z = 0$, but identical to the original string $x$ otherwise. Then, $wt(y) = wt(x) - 1$, $C_1^{y} = C_1^{x} - \{z\}$, and $C_0^{y} = C_0^{x} \cup \{z\}$.
From Eq.~(\ref{final_eigenvalue_eqn2}),
\begin{align}
& \lambda_y - \lambda_x =-2\left(\sum_{t \in C_0^{x}} J_{z_1t_1} \right) + 2\left(\sum_{s \in C_1^{y}} J_{z_1s_1} \right) \nonumber \\ 
& -2\left(\sum_{\substack{s \in C_1^{x} \\ s \neq z}} J_{z_1s_1} \right)
+4\left(\sum_{\substack{n \in C_0^{y} \\ n \neq z}}J_{z_1n_1}\right) + 2J_{z_1z_2} \nonumber \\
&= 2\left(\sum_{t \in C_0^{x}} J_{z_1t_1} \right) + 2J_{z_1z_2} \label{hamming_distance1}
\end{align}
\fi

We will use the circuit in Figure~\ref{fig:parity} to implement parity.
\begin{figure*}
\begin{center}
\begin{quantikz} [column sep=0.5cm, row sep={0.8cm,between origins}, font=\small]
\lstick{$x_1$} &\qw & \gate[2]{E} & \gate[wires = 7, nwires = 3]{U} & \qw & \qw & \qw & \qw & \qw & \qw & \qw & \gate[wires = 7, nwires = 3]{U'} & \gate[wires = 2]{E^\dagger} & \qw & \rstick{} \qw \\
\lstick{\ket{0}} & \qw &  &  & \qw & \qw & \qw & \qw & \qw & \qw & \qw &  & &\qw & \rstick{\ket{0}}\qw \\
\lstick{} & & \lstick{$\vdots$} & & & & & & & & & &\lstick{$\vdots$} & & \\
\lstick{$x_{p-1}$} & \qw & \gate[2]{E} &  & \qw & \qw & \qw & \qw & \qw & \qw & \qw &  & \gate[wires = 2]{E^\dagger} &\qw & \rstick{} \qw \\
\lstick{\ket{0}} & \qw &  &  & \qw & \qw & \qw & \qw & \qw & \qw & \qw &  & & \qw & \rstick{\ket{0}}\qw \\
\lstick{$x_p$} & \gate{H} & \gate[2]{E} &  & \gate[wires=2]{E^\dagger} & \gate{V} & \gate{H} & \ctrl{2}  & \gate{H} & \gate{V^\dagger} & \gate[wires=2]{E} &  & \gate[wires = 2]{E^\dagger} & \gate{H} & \rstick{} \qw \\
\lstick{\ket{0}} & \qw &  &  & \qw & \qw & \qw & \qw & \qw & \qw & \qw &  & & \qw  & \rstick{\ket{0}}\qw \\
& \qw & \qw & \qw & \qw & \qw & \qw & \targ{} & \qw & \qw & \qw & \qw & \qw & \qw  & \rstick{}\qw
\end{quantikz}
\end{center}
\caption{Circuit to implement parity (up to a global phase factor) using spin-exchange interactions.  The $p$ qubits carrying the input string $x$ are the \emph{input qubits}, and the last qubit is the target.  The unitary $U'$ matches $U^\dagger$ (up to a global phase factor) on the subspace spanned by encoded inputs (see Section~\ref{sec:U-dagger}).}\label{fig:parity}
\end{figure*}
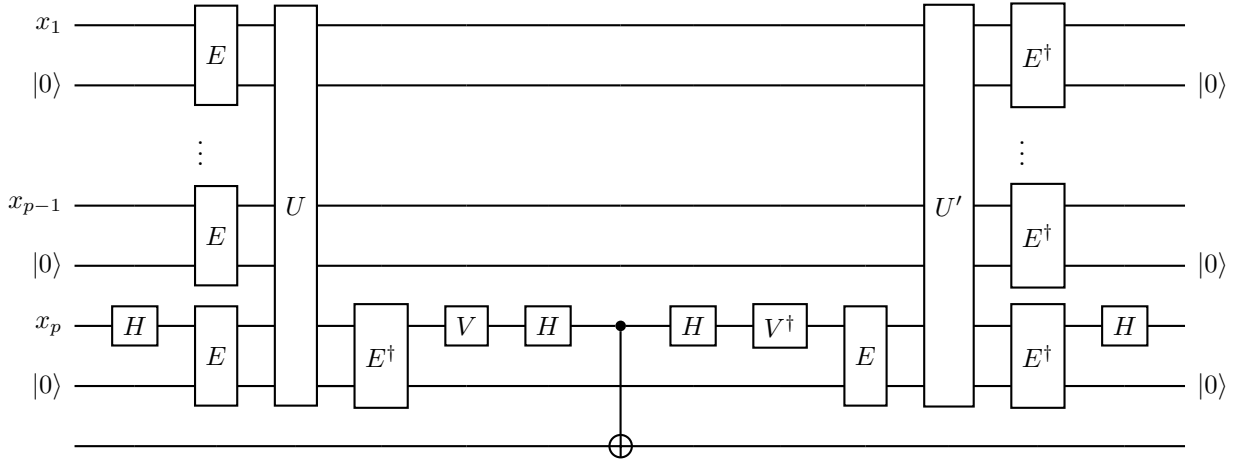
This circuit is similar to one used previously to implement parity where all interqubit couplings were equal~\cite{Fenner_Zhang}.  By running the current circuit, we derive constraints on the couplings sufficient to implement parity, revealing that a wide variety of unequal couplings suffice.

Running the circuit from left to right, we first apply the Hadamard gate to the $p^{th}$  unencoded input qubit, which we will call the \emph{active qubit}.  Then we couple each input qubit with an ancilla qubit set to $\ket{0}$ and apply the encoding unitary $E^{\otimes p}$. The resulting state is given by
\begin{align}
& E (\ket{x_1}\otimes \ket{0}) \otimes \cdots \otimes E(H\ket{x_p} \otimes \ket{0}) \nonumber \\
&= \ket{\psi_{x_1}} \otimes \cdots \otimes \ket{\psi_{x_{p-1}}} \otimes \left(\frac{E\ket{00} + (-1)^{x_p}E\ket{10}}{\sqrt2}\right) \nonumber \\
&= \ket{\psi_{x_1}} \otimes \cdots \otimes \ket{\psi_{x_{p-1}}} \otimes \left(\frac{\ket{\psi_0} + (-1)^{x_p}\ket{\psi_1}}{\sqrt2}\right) \nonumber \\
&= \frac{\ket{\psi_{x_1}} \otimes \cdots \otimes \ket{\psi_{x_{p-1}}} \otimes \ket{\psi_0}}{\sqrt2} \ + \nonumber \\
& \qquad \frac{(-1)^{x_p}\ket{\psi_{x_1}} \otimes \cdots \otimes \ket{\psi_{x_{p-1}}} \otimes \ket{\psi_1}}{\sqrt2}.
 \label{state_before_hamiltonian}
\end{align}
Applying the unitary operator $U = e^{-iTH_{g}}$ to the state given by Eq.~(\ref{state_before_hamiltonian}), we get
\begin{align}
\frac{e^{-iT\delta_u}}{\sqrt2}\ket{u_L} + \frac{(-1)^{x_p}e^{-iT\delta_v}}{\sqrt2}\ket{v_L}, \label{applied_Hamiltonian}
\end{align}
where $u = x_1x_2\ldots x_{p-1}0$ and $v = x_1x_2\ldots x_{p-1}1$.
The state shown above can be re-written as 
\begin{align}
\ket{\psi_{x_1}} &\otimes \cdots \otimes \ket{\psi_{x_{p-1}}} \nonumber\\
&\mbox{} \otimes \left(\frac{e^{-iT\delta_u}\ket{\psi_0} + (-1)^{x_p}e^{-iT\delta_v}\ket{\psi_1}}{\sqrt 2}\right)\label{state_after_encoding}
\end{align}
Then, $wt(u) = wt(v) - 1$, $C_1^{u} = C_1^{x} - \{p\}$, and $C_0^{u} = C_0^{x} \cup \{p\}$.  From %Eq.~(\ref{hamming_distance1}),
Eq.~(\ref{final_eigenvalue_eqn2}),
\begin{align}
\lambda_u - \lambda_v = 2\left(\sum_{t \in C_0^{v}} J_{p_1t_1} \right) + 2J_{p_1p_2}
\end{align}
The eigenvalue $\delta_v$ of the Hamiltonian $H_{g}$ for the string $v$ is given by
\begin{align}
\delta_v &= -\lambda_v + g(p-wt(v)) \nonumber \\
&= 2\sum_{t \in C_0^{v}} J_{p_1t_1} + 2J_{p_1p_2} - \lambda_u + g(p-wt(u) - 1) \nonumber \\
&= 2\sum_{t \in C_0^{v}} J_{p_1t_1} + 2J_{p_1p_2} -g - \lambda_u + g(p-wt(u)) \nonumber \\
&= \delta_u + 2\sum_{t \in C_0^{v}} J_{p_1t_1} + 2J_{p_1p_2} -g \nonumber \\
&= \delta_u + c(v) + c(g,p), \label{delta_v_in_u}
\end{align}
where $c(v) = 2\sum_{t \in C_0^{v}} J_{p_1t_1}$ and $c(g,p) = \\ 2J_{p_1p_2} -g$.
So, using Eq.~(\ref{delta_v_in_u}) in Expression~(\ref{state_after_encoding}), the state of the $p^{th}$ pair is equal to 
\begin{align}
& e^{-iT\delta_u}\;\frac{\ket{\psi_0} + (-1)^{x_p}e^{-iTc(v)}e^{-iTc(g,p)}\ket{\psi_1}}{\sqrt 2} \nonumber 
\end{align}
Using $E^\dagger$ to decode the $p^{th}$ pair of qubit gives
\begin{align}
& e^{-iT\delta_u}\left(\frac{\ket{0} + (-1)^{x_p}e^{-iTc(v)}e^{-iTc(g,p)}\ket{1}}{\sqrt 2} \right) \otimes \ket{0} \nonumber \\
\end{align}
Applying the unitary operator $W$ = $\begin{bmatrix}
1 & 0 \\
0 & e^{iTc(g, p)}
\end{bmatrix}$ to the $(n-1)^{th}$ qubit gives
\begin{align}
& e^{-iT\delta_u}\left(\frac{\ket{0} + (-1)^{x_p}e^{-iTc(v)}\ket{1}}{\sqrt 2} \right) \otimes \ket{0} \nonumber \\
& e^{-iT\delta_u}\left(\frac{\ket{0} + e^{-i(\pi x_p + Tc(v))}\ket{1}}{\sqrt 2} \right) \otimes \ket{0} \nonumber \\
\end{align}
Now, applying a Hadamard to the $(n-1)^{th}$ qubit gives
\begin{align}
& e^{-iT\delta_u}\left(\frac{1+ e^{-i(\pi x_p + Tc(v))}}{2}\right)\ket{0} + e^{-iT\delta_u}\left(\frac{1- e^{-i(\pi x_p + Tc(v))}}{2}\right)\ket{1} \nonumber \\
&= e^{-iT\delta_u}e^{-i(\pi x_p + Tc(v))/2}\left(\cos\left(\frac{\pi x_p + Tc(v)}{2}\right)\ket{0} 
+ i\sin\left(\frac{\pi x_p + Tc(v)}{2}\right)\ket{1} \right) \nonumber \\
&= e^{-i(T\delta_u+\phi)}\left(\cos(\phi)\ket{0} + i\sin(\phi)\ket{1}\right) \,, \label{before_cnot}
\end{align}
where $\phi = \left(\frac{\pi x_p + Tc(v)}{2}\right)$.

We want $\sin(\phi) = 0$ when the parity is $0$, and $\cos(\phi) = 0$ when the parity is $1$.
Fix a string $x$ such that $x_1 = 0$ and $x_2\ldots y_{p-1} = 1 \ldots 1$. Then, $c(v) = 2J_{p_1,1_1}$. 
For all real $a, b, \alpha$ with $\alpha \neq 0$, we use $a \equiv_\alpha b$ to mean $(a-b)/\alpha$ is an integer.

If $p$ is odd and parity is $0$, then $x_p = 1$. Then, $\phi = \frac{\pi}{2} + TJ_{p_1,1_1}$. We want $\sin(\phi) = 0 \implies
\phi \equiv_\pi 0 \implies TJ_{p_1,1_1} \equiv_\pi  \frac{\pi}{2}$.

If $p$ is odd and parity is $1$, then $x_p = 0$. Then, $\phi = TJ_{p_1,1_1}$. We want $\cos(\phi) = 0 \implies \phi \equiv_\pi \frac{\pi}{2} \implies TJ_{p_1,1_1} \equiv_\pi \frac{\pi}{2}$. 
\\

If $p$ is even, apply an $X$ operator to the $(n-1)^{th}$ qubit, whose state is given in Eq.~(\ref{before_cnot}). Then the state is given by 
\begin{align}
& e^{-i(T\delta_u+\phi)}\left(\cos(\phi)\ket{1} + i\sin(\phi)\ket{0} \right) \otimes \ket{0}. \label{after X}
\end{align}

We want $\cos(\phi) = 0$ when the parity is $0$, and $\sin(\phi) = 0$ when the parity is $1$. 

If $p$ is even and parity is $0$, then $x_p = 0$. Then, $\phi = TJ_{p_1,1_1}$. Thus, $\cos(\phi) = 0 \implies
\phi \equiv_\pi \frac{\pi}{2} \implies TJ_{p_1,1_1} \equiv_\pi  \frac{\pi}{2}$.

If $p$ is even and parity is $1$, then $x_p = 1$. Then, $\phi = \frac{\pi}{2} + TJ_{p_1,1_1}$. Thus, $\sin(\phi) = 0 \implies
\phi \equiv_\pi 0 \implies TJ_{p_1,1_1} \equiv_\pi  \frac{\pi}{2}$.

Therefore, in all cases, we get the same constraint $TJ_{p_1,1_1} \equiv_\pi  \frac{\pi}{2}$. Notice that we could change the string $x$ such that $x_e = 0$ for a given $e \in [1, p-1]$, and $x_f = 1$ for any $f \in [1, p-1]$ and $e \neq f$. So, 
\begin{equation} \label{looser_constraint}
TJ_{p_1,e_1} \equiv_\pi  \frac{\pi}{2}
\end{equation} for any $e \in [1, p-1]$. 
Under this constraint,
\begin{align}
\frac{Tc(v)}{2} & =  T\sum_{t \in C_0^v} J_{p_1t_1} \equiv_{\pi} TJ_{p_1,1_1}(p-1-wt(u)). \label{c(v)_equiv_eqn}
\end{align}
So, $\phi \equiv_\pi \frac{\pi}{2}x_p + \frac{\pi}{2}(p-1-wt(u)) \equiv_\pi \frac{\pi}{2} (x_p + wt(u)) + \frac{\pi}{2}(p-1) \equiv_\pi \frac{\pi}{2} wt(x) + \frac{\pi}{2}(p-1)$. For odd $p$ and even parity, $\phi \equiv_\pi 0$. So, from Eq.~(\ref{before_cnot}), the state of the $(n-1)^{th}$ qubit is proportional to $\ket{0}$, and applying the CNOT gate with the $(n-1)^{th}$ qubit as the control, the parity bit can be copied onto a fresh ancilla qubit initally set to $\ket{0}$. Similarly, one can verify for odd $p$ and odd parity. For even $p$, with the application of $X$ operator, one can verify that the circuit shown in Figure~\ref{fig:parity} works for both even and odd parity using Eq.~(\ref{after X}).

Also, since $HX = ZH$, we can choose to apply $Z$ operator to the $(n-1)^{th}$ qubit before applying the Hadamard in case $p$ is even. Set $q \equiv_2 p+1$. So, we can apply the unitary operator $V = WZ^{q}$ , which is given by $\begin{bmatrix}
1 & 0 \\
0 & (-1)^{q}e^{iTc(g, p)}
\end{bmatrix}$, to the $(n-1)^{th}$ qubit as shown in Figure~\ref{fig:parity}. Lastly, we apply a CNOT gate to copy the parity bit to a fresh ancilla qubit. Notice that the extra phase factor $e^{-iT\delta_u}e^{-i\frac{(\pi x_p + Tc(v))}{2}}$ gets cancelled by the rest of the circuit, which uses a unitary $U'$ (see Section~\ref{sec:U-dagger}, below). Finally, we apply a bank of Hadamard gates $H^{\otimes (n+1)}$ on each side of the circuit shown in Figure~\ref{fig:parity} to obtain the fanout gate.

A more flexible implementation of parity would allow \emph{any} of the input qubits to be the active qubit, not necessarily the $p^{th}$.  For this, the constraints on the coupling coefficients would require
\begin{equation} \label{stronger_constraint}
TJ_{f_1,\ell_1} \equiv_\pi \frac{\pi}{2}
\end{equation}
for every distinct $f,\ell \in [1, p]$. We will use these stronger constraints in the following section to implement $U'$ that matches the inverse $U^\dagger$ on the subspace spanned by the encoded inputs $\ket{x_L}$.

\section{Inverse of the unitary $U$} \label{sec:U-dagger}
To make the circuit clean, we must implement a unitary $U'$ that matches $U^\dagger$ on the subspace spanned by the $\ket{x_L}$, up to a global phase factor.  The naturally easiest way to do this is to evolve the same Hamiltonian $H_g$ for some time $T'$.  This will impose additional constraints on the couplings.  We solve the following equation for an arbitrary $p$-bit string $x$ and some real $\theta$:
\begin{equation} \label{inverse_unitary}
e^{-iT^\prime\delta_x}e^{-iT\delta_x} \ket{x_L} = e^{-i\theta} \ket{x_L}.
\end{equation}
where $\theta$ is some real number independent of $x$.
Let $u$ denote the Hamming weight of the string $x$. So, Eq.~(\ref{inverse_unitary}) becomes
%, and $T^\prime$ represents the time of evolution the altered Hamiltonian $H_{g^\prime}$ needs to be turned on for to implement $U^\dagger$. 
\begin{align}
&(T^\prime + T)\delta_x \equiv_{2\pi} \theta \nonumber \\
&(T^\prime + T)(-\lambda_x + g(p-u)) \equiv_{2\pi} \theta \nonumber \\
%& T^\prime(-\lambda^\prime_x + g^\prime(p-wt(x))) + Tgp -T(\lambda_x + gwt(x)) \equiv_{2\pi} \theta \nonumber \\
&(T^\prime + T)(\lambda_x + gu) \equiv_{2\pi} (T+T')gp - \theta = \theta'\;,\label{phase_independent_x}
\end{align}
where $\theta'$ does not depend on $x$.

We will assume the following additional constraints:
\begin{enumerate}
\item $T = kT'$ for some odd integer $k$.
\item $T'J_{f_1,\ell_1} \equiv_\pi \frac{\pi}{2}$ for every distinct $f,\ell \in [1, p]$.
\item For each $f \in [1,p]$, \ $T'J_{f_1f_2} \equiv_{\pi} 0$. 
\item The parameter $g$ controlling the external magnetic field 
is such that $T'g \equiv_{\pi} 0$.
\end{enumerate}

Under these additional constraints, we get
\begin{align*} \label{T_lambda2}
& (T+T')(\lambda_x + gu) \\
& \equiv_{2\pi} (k+1)T'\lambda_x + (k+1)T'gu \\ 
%& \qquad \text{(using constraint no. 1)} \\
& \equiv_{2\pi} (k+1)\pi \left( {p \choose 2} + {p-u \choose 2}\right) + (k+1)T'c_x \\ 
& \qquad + (k+1)T'gu \\
& \equiv_{2\pi} (k+1)\pi \left( {p \choose 2} + {p-u \choose 2}\right) + (k+1)\pi k' \\
& \qquad + (k+1)\pi k''u \\
& \equiv_{2\pi} 0.
\end{align*}
The first line of congruence follows from constraint no. 1 between $T$ and $T'$. The second line of congruence follows from constraint no. 2 on the external couplings, and the quantity $c_x$ that appears is the term that only includes internal couplings and depends on $x$. Precisely,
\begin{equation} \label{eqn_cx}
c_x = \left(\sum_{m \in C_0^{x}} J_{m_1m_2} \right) - \left(\sum_{r \in C_1^{x}} J_{r_1r_2}\right).
\end{equation} 
See Eq.~(\ref{final_eigenvalue_eqn2}) for clarity. The third line of congruence follows from the last two constraints on the internal couplings and the external magnetic field. The integers $k'$ and $k''$ arise from the last two constraints. Thus, under these additional constraints, Eq.~(\ref{phase_independent_x}) is satisfied.

\section{Implementing the $\text{Mod}_q$ gate}
In this section, we will show how to directly implement $\Mod_q$ gate using the Hamiltonian $H_g$ given by Eq.~(\ref{Hamiltonian}) for $q \geq 2$. It is already known that $\Mod_q$ gates for all $q \geq 2$ are constant-depth equivalent to each other \cite{Green}, however our approach here is much more direct. The $\Mod_q$ is a classical gate that acts on $r$ control bits and a target bit. The target bit is flipped iff the Hamming weight of the control bits is not a multiple of $q$. We will simulate a more powerful version of this gate, the \textit{generalized} $\Mod_q$ \textit{gate}, which has $r$ control bits and $q-1$ target bits $t_1, \ldots, t_{q-1}$. If $w$ is the Hamming weight of the control bits, then the target bits $t_1, \ldots, t_i$ are all flipped, where $i = w \tpmod {q}$, and the other target bits are left alone. Figure~\ref{fig:generalized_mod} shows how to simulate a (standard) $\Mod_q$ gate with a circuit using two generalized $\Mod_q$ gates and a CNOT gate.

\begin{figure*}
\begin{center}
\begin{quantikz} [column sep=0.5cm, row sep={0.8cm,between origins}, font=\small]
        &  \ctrl{3} & \qw \\
\vdots  &           & \vdots   \\
        & \ctrl{5}  & \qw \\
& & \\
& & \\
& & \\
& & \\
& \gate[1]{q}  & \qw
\end{quantikz}
=
\begin{quantikz} [column sep=0.5cm, row sep={0.8cm,between origins}, font=\small]
                   & \ctrl{5} & \qw & \ctrl{5} & \qw \\
           \vdots  &          &     &          & \vdots \\
                   & \ctrl{4} & \qw & \ctrl{4} & \qw \\
\lstick{$\ket{0}$}   & \gate[4,nwires={3}]{q} & \ctrl{4} & \gate[4,nwires={3}]{q} & \rstick{$\ket{0}$} \qw \\
\lstick{$\ket{0}$} & \qw & \qw & \qw & \rstick{$\ket{0}$} \qw \\
\vdots & & & & \lstick{$\vdots$} & \\
\lstick{$\ket{0}$} & \qw & \qw & \qw & \rstick{$\ket{0}$} \qw \\
& \qw & \targ{} & \qw &  \qw 
%\lstick{} & \ctrl{1} & \rstick{} \qw \\
%\lstick{} & \gate[1]{q}  & \rstick{} \qw
%\lstick{}  & \lstick{$\vdots$} & \lstick{$\vdots$}   \\
%\lstick{} & \ctrl{1} & \rstick{} \qw \\
%\lstick{} & \gate[1]{q}  & \rstick{} \qw
%\lstick{$x_1$} &\qw & \gate[2]{E} & \gate[wires = 7, nwires = 3]{U} & \qw & \qw & \qw & \qw & \qw & \qw & \qw & \gate[wires = 7, nwires = 3]{U^\dagger} & \gate[wires = 2]{E^\dagger} & \qw & \rstick{} \qw 
%\lstick{\ket{0}} & \qw &  &  & \qw & \qw & \qw & \qw & \qw & \qw & \qw &  & &\qw & \rstick{\ket{0}}\qw \\
%\lstick{} & & \lstick{$\vdots$} & & & & & & & & & &\lstick{$\vdots$} & & \\
%\lstick{$x_{p-1}$} & \qw & \gate[2]{E} &  & \qw & \qw & \qw & \qw & \qw & \qw & \qw &  & \gate[wires = 2]{E^\dagger} &\qw & \rstick{} \qw \\
%\lstick{\ket{0}} & \qw &  &  & \qw & \qw & \qw & \qw & \qw & \qw & \qw &  & & \qw & \rstick{\ket{0}}\qw \\
%\lstick{$x_p$} & \gate{H} & \gate[2]{E} &  & \gate[wires=2]{E^\dagger} & \gate{V} & \gate{H} & \ctrl{2}  & \gate{H} & \gate{V^\dagger} & \gate[wires=2]{E} &  & \gate[wires = 2]{E^\dagger} & \gate{H} & \rstick{} \qw \\
%\lstick{\ket{0}} & \qw &  &  & \qw & \qw & \qw & \qw & \qw & \qw & \qw &  & & \qw  & \rstick{\ket{0}}\qw \\
%\lstick{\ket{0}} & \qw & \qw & \qw & \qw & \qw & \qw & \targ{} & \qw & \qw & \qw & \qw & \qw & \qw  & \rstick{}\qw
\end{quantikz}
\end{center}
\caption{Simulating a standard $\Mod_q$ gate with generalized $\Mod_q$ gates. There are $r$ control qubits and the ancillae qubits on the right are the target qubits of the generalized $\Mod_q$ gate which are labeled $t_1, t_2, \ldots, t_{q-1}.$} \label{fig:generalized_mod}
\end{figure*}
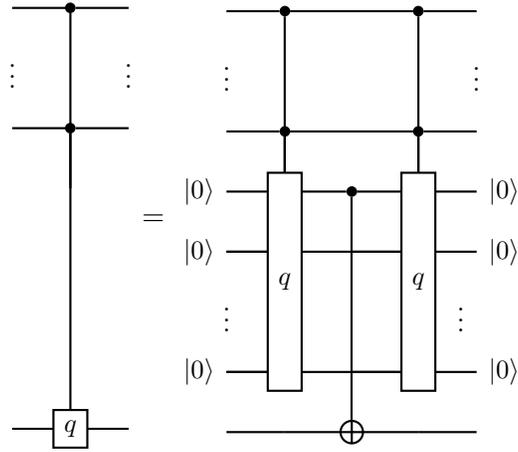

\begin{figure*}
\begin{center}
\begin{quantikz} [column sep=0.5cm, row sep={0.7cm,between origins}, font=\small]
\ghost{X}& \ctrl{2} & \qw \\
\vdots  &          & \vdots \\
        & \ctrl{4} & \qw \\
        &          & \\
        &          & \\
        &          & \\
        & \gate[3,nwires={2}]{q} &\qw \\
\vdots  &          & \vdots \\
        &          & \qw
\end{quantikz}
=
\begin{quantikz} [column sep=0.5cm, row sep={0.7cm,between origins}, font=\small]
&\gate[6,nwires={2,5}]{E}&\gate[6,nwires={2,5}]{U}&\gate[6,nwires={2,5}]{E^\dagger}&\qw&\qw&\qw&\qw&\qw&\gate[6,nwires={2,5}]{E}&\gate[6,nwires={2,5}]{U'}&\gate[6,nwires={2,5}]{E^\dagger}&\qw \\
\vdots &&&&&&\vdots&&&&&& \vdots \\
&&&&\qw&\qw&\qw&\qw&\qw&&&& \qw \\
\lstick[wires=3]{$\ket{\mathbb{A}}$}  & &&&\gate[3,nwires={2}]{R} &\ctrl{3}&\qw&\qw&\gate[3,nwires={2}]{R^\dagger}&\qw&&& \qw\rstick[wires=3]{$\ket{\mathbb{A}}$} \\
&\vdots&&&&&\ddots&&&&&& \vdots& \\
&&&&\qw&\qw&\qw&\ctrl{3}&&&&& \qw \\
& \qw & \qw & \qw & \qw &\targ{}& \qw & \qw & \qw & \qw & \qw & \qw &\qw \\
& &\vdots &&&&\ddots&&&& \vdots&& \\
&\qw & \qw & \qw & \qw &\qw& \qw & \targ{} & \qw &\qw &\qw &\qw & \qw
\end{quantikz}
\end{center}
\caption{Circuit to implement generalized $\Mod_q$ gate using spin-exchange interactions} \label{fig: generalized_Mod_q_implementation}
\end{figure*}
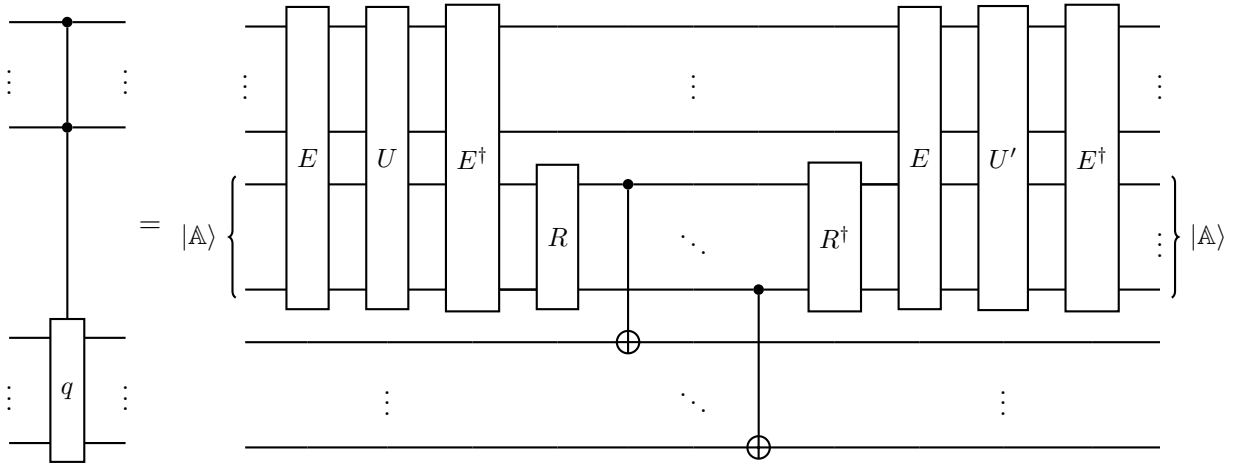

We prepare the $q-1$ ancillae qubits into the state
\begin{align}
\ket{\mathbb{A}} = \sum_{j=0}^{q-1}h_j \ket{A_j},
\end{align}
where $\ket{A_j} = \ket{1^j0^{q-1-j}}$ and $|h_j|^2 = \frac{1}{q}$ for all $j$.
Let $\mathcal{H}$ denote the Hilbert space of a qubit. Consider a $p-$qubit computational basis state $\ket{x}$ for $p \geq 2$ with Hamming weight $u$. We encode the state $\ket{x}\otimes\ket{\mathbb{A}}$ with our encoder $E$ by using one ancilla qubit for each data qubit in the state $\ket{x}\otimes\ket{\mathbb{A}}$ to produce the encoded state $\ket{(x\mathbb{A})_L} \in \{\ket{\psi_0}, \ket{\psi_1}\}^{\otimes p+q-1}$. Then, we turn on the Hamiltonian $H_g = -J^2 + gJ_z$ for time period $T$ such that 
\begin{equation} \label{mod_q_constraint}
2TJ_{f_1\ell_1} \equiv_{2\pi} \frac{2\pi k}{q},
\end{equation}
for every distinct $f,\ell \in [1,p]$ and for an integer $k > 0$ co-prime to $q$. Notice that this new constraint recovers the constraint given in Eq.~(\ref{stronger_constraint}) that we obtained for parity $(q=2)$. Under this new constraint on the coupling strengths and Eq.~(\ref{final_eigenvalue_eqn2}), notice that 
\begin{equation} \label{T_lambda}
T\lambda_x \equiv_{2\pi} \frac{2\pi k}{q} {p \choose 2} + \frac{2\pi k}{q} {p-u \choose 2} + Tc_x,
\end{equation}
where $c_x$, as given in Eq.~(\ref{eqn_cx}), is the term that only includes internal couplings and depends on $x$. 
%Precisely,
%$$
%c_x = \left(\sum_{m \in C_0^{x}} J_{m_1m_2} \right) - \left(\sum_{r \in C_1^{x}} J_{r_1r_2}\right)
%$$
So, $U\ket{(x\mathbb{A})_L} = e^{-iTH_g}\ket{(x\mathbb{A})_L} = $ 
\begin{align} 
 & \sum_{j=0}^{q-1} h_j e^{-iT(-\lambda_{xA_j} + g(p+q-1 - u -j))}\ket{(xA_j)_L} \label{encoded_modq}
\end{align}
Also, we can write
\begin{align}
T\lambda_{xA_j} 
 & \equiv_{2\pi} \frac{2\pi k}{q} {z \choose 2} 
+ \frac{2\pi k}{q} {z-u-j \choose 2} + Tc_{xA_j}, \nonumber
\end{align}
for $z = p+q-1$.
Notice that the string $xA_j$ is the concatenation of the strings $x$ and $A_j$. So, $c_{xA_j} = c_x + c_{A_j}$.
Next, we decode the state shown in Eq.~(\ref{encoded_modq}) using $(E^\dagger)^{\otimes z}$. Doing so will restrict the information necessary to compute the residue of modular $q$ into the $q-1$ ancilla qubits. The state of the $q-1$ qubits is then given by 
\begin{align}
\ket{\Phi_u} = \sum_{j=0}^{q-1} h_je^{-iT(-\lambda_{xA_j} + g(z - u-j))}\ket{1^j0^{q-j-1}}
\end{align}
Notice that all the states of the form $\ket{\Phi_u}$ lie in a $q$-dimensional subspace $\mathcal{H}^\prime$ of $\mathcal{H}^{\otimes (q -1)}$, spanned by $\{\ket{A_j} | 0 \leq j < q\}$. For any other $p$-qubit computational basis state $\ket{y}$ with Hamming weight $v$, we want to show that $\braket{\Phi_v}{\Phi_u} = 0$ when $u \not\equiv_{q} v$. So,
\begin{align}
& \braket{\Phi_v}{\Phi_u} = \sum_{j=0}^{q-1} |h_j|^2 e^{iT(\lambda_{xA_j} - \lambda_{yA_j} + gu - gv)} \nonumber \\
& \propto \ \frac{1}{q} \sum_{j=0}^{q-1} e^{i(T\lambda_{xA_j} - T\lambda_{yA_j} )} \nonumber \\
& = \frac{1}{q} \sum_{j=0}^{q-1} e^{i\left(\frac{2\pi k}{q} \left({z-u-j \choose 2} - {z-v-j \choose 2}\right) + Tc_{xA_j} - Tc_{yA_j} \right)} \nonumber \\
& = \frac{1}{q} \sum_{j=0}^{q-1} e^{i\left(\frac{2\pi k}{q} {z-u-j \choose 2} - \frac{2\pi k}{q} {z-v-j \choose 2} + Tc_{x} - Tc_{y} \right)} \nonumber \\
& \propto \ \frac{1}{q} \sum_{j=0}^{q-1} e^{i\frac{2\pi k}{q}\left( {z-u-j \choose 2} - {z-v-j \choose 2} \right)} \nonumber \\
& \propto \ \frac{1}{q} \sum_{j=0}^{q-1} e^{\frac{i2\pi jk}{q}(u - v)} \nonumber \\
& = \delta_{u \tpmod{q}, \ v \tpmod{q}},
\end{align}
 where $\delta_{w,o}$ is the Kronecker delta.
 
Thus, there exists an orthonormal basis $\{\ket{\phi_{u \tpmod {q}}} | \ u \in \mathbb{N} \}$ of $\mathcal{H}^\prime$ such that each state $\ket{\Phi_u} = e^{i\theta_u}\ket{\phi_{u \tpmod {q}}}$ for some $\theta_u \in \mathbb{R}$. Then, there exists a unitary operator $R$ that maps each state $\ket{\phi_j}$ to the state $\ket{A_j}$. We then use a CNOT to copy the $j^{th}$ ancilla qubit value into the $j^{th}$ target qubit $t_j$. Lastly, we apply the inverse of the previous computations to get rid of conditional phase factors, as shown in Figure~\ref{fig: generalized_Mod_q_implementation}.

In order to implement $U'$ that matches the inverse $U^\dagger$ (up to a global phase factor) on the subspace spanned by the encoded inputs $\ket{x_L}$, we impose similar additional constraints on the Hamiltonian $H_g$ as before (see Section~\ref{sec:U-dagger}):
\begin{enumerate}
%\item $T = k'''T'$, where $k'''$ is any integer such that $k''' \equiv_q q-1$.
%\item By magnitude, the reciprocal of each inter-pair coupling $J_{f_1l_1}$ is congruent modulo $2\pi$ to an integer $k_{fl}$ co-prime to $q$. That is, $\frac{1}{J_{f_1l_1}} \equiv_{2\pi} k_{fl}$ for every distinct $f,l \in [1,p]$.
\item For each $f \in [1,p], TJ_{f_1f_2} \equiv_{2\pi} \frac{2\pi k^\prime}{q}$ for some integer $k^\prime$. 
%That is, all the intra-pair couplings are also subject to the same constraint as the inter-pair couplings (shown in Eq.~(\ref{mod_q_constraint})).
\item The parameter $g$ controlling the external magnetic field
%and the sum of the intra-pair couplings (denoted by $c_x$) are each 
is such that $Tg \equiv_{2\pi} \frac{2\pi k^{\prime \prime}}{q}$ for some integer $k^{\prime \prime}$.
\end{enumerate}
Under these additional constraints, in order to implement $U'$, we can evolve the Hamiltonian $H_g$ for time of evolution $T' = k'''T$, where $k'''$ is any integer such that $k''' \equiv_q q-1$.

\section{Conclusion}
In this work, we give a constant-depth implementation of fanout and $\Mod_q (q \geq 2)$ gate by evolving the system of $n (\geq 4)$ physical qubits according to a Heisenberg Hamiltonian $H_g$ (Eq.~(\ref{Hamiltonian})) and provide the exact constraints on the couplings among the qubits that allow for this implementation. Our results assume that the Heisenberg Hamiltonians of the form $H_g$ are easy enough to implement in lab so as to easily produce the respective unitaries required for parity and $\Mod_q$ gates. 

There are some interesting questions still left open in our current work. To realize the quantum fanout gate, we encode each logical qubit into the spin states $\ket{\psi_0}$ and $\ket{\psi_1}$ of two physical qubits, yielding a physical-to-logical qubit number ratio as 2. It would be interesting to see if one can group multiple logical qubits together and produce a compressed encoding that has physical-to-logical qubit number ratio as close to unity as possible, perhaps by encoding the logical qubits into some other spin states. Producing a compressed encoding would allow a reduction in the number of required ancilla qubits (and hence, the circuit width) for implementing the fanout gate, but may require a larger encoding circuit that acts on distant qubits rather than just the nearest-neighbors. It would also be interesting to see what other multi-qubit entangling gates can be implemented in constant- or shallow-depth quantum circuits using evolution of Hamiltonians that naturally arise in physical systems of interest. We used an XXX isotropic Heisenberg Hamiltonian. However, in real physical systems, the spin-spin interactions can be anisotropic due to surface and interface effects, along with spin-orbit coupling \cite{Lidar}. So, it would be interesting to derive the constraints on the coupling strengths via anisotropic XXZ or XYZ models. One could also consider Hamiltonians where the interactions among the qubits satisfy inverse power laws to explore the restrictions on the spatial arrangements of spin qubits in certain experimental systems, for example, the location of spin qubits in semiconductor quantum dots, to implement fanout.
\bibliographystyle{IEEEtran}
\bibliography{new_heisenberg_modified}

\appendix
\section{Solving eigenvalue in terms of coupling strengths (Eq.~(\ref{eigen_equation}))} \label{lambda_solution}
Without loss of generality, throughout the proof, we assume that the p-bit string $x$ starts with the $1$'s followed by the $0$'s, i.e., $C_1^x = \{1, \ldots, w \}$ and $C_0^x = \{w+1, \ldots, p \}$ where $w$ is the Hamming weight of $x$. We also consider a more general encoder $E$ that encodes as follows: 
\begin{align}
\ket{\psi_0} &:= E\ket{00} = \ket{00},   \nonumber \\
\ket{\psi_1} &:= E\ket{10} = \alpha\ket{01} + \beta\ket{10} \text{for some} \ \alpha, \beta \in \mathbb{C} \ \ \text{such that} \ |\alpha|^2 + |\beta|^2 = 1.
\end{align}
The sum $\sum_{1 \leq i < j \leq n} J_{ij} \SWAP_{ij}$ in Eq.~(\ref{modified_J2}) can be split into four smaller sums:
\begin{enumerate}
\item 
$ \sum_{1 \leq r < s \leq w} J_{r_1s_1} \SWAP_{r_1s_1} + J_{r_1s_2} \SWAP_{r_1s_2} + J_{r_2s_1} \SWAP_{r_2s_1} + J_{r_2s_2} \SWAP_{r_2s_2}  $ 
\vspace{0.15cm}
\item $ \sum_{1 \leq r \leq w, 2w+1 \leq t \leq n} J_{r_1t} \SWAP_{r_1t} + J_{r_2t} \SWAP_{r_2t}  $
\vspace{0.15cm}
\item $ \sum_{1 \leq r \leq w} J_{r_1r_2} \SWAP_{r_1r_2} $
\vspace{0.15cm}
\item $ \sum_{2w+1 \leq t_1 < t_2 \leq n} J_{t_1t_2} \SWAP_{t_1t_2} $
\end{enumerate}
We evaluate each individual sum applied to the encoded state $\ket{x_L}$. Starting with the first sum, we get
\begin{align} 
\sum_{1 \leq r < s \leq w} & \bigl(J_{r_1s_1} \SWAP_{r_1s_1} + J_{r_1s_2} \SWAP_{r_1s_2} + J_{r_2s_1} \SWAP_{r_2s_1} + J_{r_2s_2} \SWAP_{r_2s_2} \bigr) \ket{\psi_1}^{\otimes w} \ket{\psi_0}^{\otimes (p-w)} \nonumber \\ 
= \sum_{1 \leq r < s \leq w} & \bigl(J_{r_1, s_1} \SWAP_{r_1, s_1} + 
J_{r_1s_2} \SWAP_{r_1s_2} + J_{r_2s_1} \SWAP_{r_2s_1} + J_{r_2s_2} \SWAP_{r_2s_2} \bigr)  \nonumber \\
& \left(\alpha^2\ket{\psi_1^{\otimes(r-1)}01\psi_1^{\otimes(s-r-1)}01\psi_1^{\otimes (w-s)}\psi_0^{\otimes (p-w)}} \right. \nonumber \\ 
& \qquad \left. + \alpha\beta\ket{\psi_1^{\otimes(r-1)}01\psi_1^{\otimes(s-r-1)}10\psi_1^{\otimes (w-s)}\psi_0^{\otimes (p-w)}} \right. \nonumber \\ 
& \qquad \left. + \alpha\beta\ket{\psi_1^{\otimes(r-1)}10\psi_1^{\otimes(s-r-1)}01\psi_1^{\otimes (w-s)}\psi_0^{\otimes (p-w)}} \right. \nonumber \\ 
& \qquad \left. + \beta^2\ket{\psi_1^{\otimes(r-1)}10\psi_1^{\otimes(s-r-1)}10\psi_1^{\otimes (w-s)}\psi_0^{\otimes (p-w)}} \right) \nonumber \\ 
= \sum_{1 \leq r < s \leq w} & J_{r_1s_1} \left(\alpha^2\ket{\cdots01\cdots01\cdots} \right. \nonumber \\ 
& \qquad \left. + \alpha\beta\ket{\cdots11\cdots00\cdots} \right. \nonumber \\  
& \qquad \left. + \alpha\beta\ket{\cdots00\cdots11\cdots} \right. \nonumber \\ 
& \qquad \left. + \beta^2\ket{\cdots10\cdots10\cdots} \right) \nonumber \\ 
& + J_{r_1s_2} \left(\alpha^2\ket{\cdots11\cdots00\cdots} \right. \nonumber \\ 
& \qquad \left. + \alpha\beta\ket{\cdots01\cdots10\cdots} \right. \nonumber \\ 
& \qquad \left. + \alpha\beta\ket{\cdots10\cdots01\cdots} \right. \nonumber \\ 
& \qquad \left. + \beta^2\ket{\cdots00\cdots11\cdots} \right) \nonumber \\ 
& + J_{r_2s_1} \left(\alpha^2\ket{\cdots00\cdots11\cdots} \right. \nonumber \\ 
& \qquad \left. + \alpha\beta\ket{\cdots01\cdots10\cdots} \right. \nonumber \\ 
& \qquad \left. + \alpha\beta\ket{\cdots10\cdots01\cdots} \right. \nonumber \\ 
& \qquad \left. + \beta^2\ket{\cdots11\cdots00\cdots} \right) \nonumber \\ 
& + J_{r_2s_2} \left(\alpha^2\ket{\cdots01\cdots01\cdots} \right. \nonumber \\ 
& \qquad \left. + \alpha\beta\ket{\cdots00\cdots11\cdots} \right. \nonumber \\  
& \qquad \left. + \alpha\beta\ket{\cdots11\cdots00\cdots} \right. \nonumber \\ 
& \qquad \left. + \beta^2\ket{\cdots10\cdots10\cdots} \right) \label{constr_0}
\end{align}

The states $\ket{\psi_1^{\otimes(r-1)}00\psi_1^{\otimes(s-r-1)}11\psi_1^{\otimes (w-s)}\psi_0^{\otimes (p-w)}}$ and \\
$\ket{\psi_1^{\otimes(r-1)}11\psi_1^{\otimes(s-r-1)}00\psi_1^{\otimes (w-s)}\psi_0^{\otimes (p-w)}}$ are linearly independent with all other states on L.H.S. and do not appear on R.H.S. of Eq.~(\ref{eigen_equation}). So, equating their coefficients to zero, we get
\begin{align}
J_{r_1s_1}\alpha\beta + J_{r_1s_2}\beta^2 + J_{r_2s_1}\alpha^2 + J_{r_2s_2}\alpha\beta = 0 \label{constr_1} \\ 
J_{r_1s_1}\alpha\beta + J_{r_1s_2}\alpha^2 + J_{r_2s_1}\beta^2 + J_{r_2s_2}\alpha\beta = 0 \label{constr_2}
\end{align}
Equating Eqs.~\ref{constr_1} and \ref{constr_2}, we get
\begin{equation}
(J_{r_1s_2} - J_{r_2s_1}) (\alpha^2 - \beta^2) = 0 \label{constr_3}
\end{equation}
Adding Eqs.~\ref{constr_1} and \ref{constr_2}, we get
\begin{equation}
2\alpha\beta(J_{r_1s_1} + J_{r_2s_2}) + (\alpha^2 + \beta^2)(J_{r_1s_2} + J_{r_2s_1}) = 0 \label{constr_4}
\end{equation}

Eq.~(\ref{constr_0}) can be rewritten as follows,
\begin{align}
= \sum_{1 \leq r < s \leq w} & \bigl(J_{r_1s_1} + J_{r_2s_2}\bigr)\alpha^2\ket{\psi_1^{\otimes(r-1)}01\psi_1^{\otimes(s-r-1)}01\psi_1^{\otimes (w-s)}\psi_0^{\otimes (p-w)}}  \nonumber \\
& + \bigl(J_{r_1s_2} + J_{r_2s_1}\bigr)\alpha\beta\ket{\psi_1^{\otimes(r-1)}01\psi_1^{\otimes(s-r-1)}10\psi_1^{\otimes (w-s)}\psi_0^{\otimes (p-w)}} \nonumber \\ 
& + \bigl(J_{r_1s_2} + J_{r_2s_1}\bigr)\alpha\beta\ket{\psi_1^{\otimes(r-1)}10\psi_1^{\otimes(s-r-1)}01\psi_1^{\otimes (w-s)}\psi_0^{\otimes (p-w)}} \nonumber \\ 
& + \bigl(J_{r_1s_1} + J_{r_2s_2}\bigr)\beta^2\ket{\psi_1^{\otimes(r-1)}10\psi_1^{\otimes(s-r-1)}10\psi_1^{\otimes (w-s)}\psi_0^{\otimes (p-w)}}  \label{constr_a}
\end{align}

By evaluating the second sum, we get
\begin{align} \nonumber
= \sum_{1 \leq r \leq w, \ 2w+1 \leq t \leq n} & \bigl(J_{r_1t} \SWAP_{r_1t} + J_{r_2t} \SWAP_{r_2t}\bigr) \ket{\psi_1}^{\otimes w} \ket{\psi_0}^{\otimes (p-w)} \\ \nonumber
= \sum_{1 \leq r \leq w, \ 2w+1 \leq t \leq n} & \bigl(J_{r_1t} \SWAP_{r_1t} + J_{r_2t} \SWAP_{r_2t} \bigr) \\ \nonumber
& \left(\alpha\ket{\psi_1^{\otimes(r-1)}01\psi_1^{\otimes(w-r)}0^{t-2w-1}00^{n-t}} \right. \\ \nonumber
& \left. + \beta\ket{\psi_1^{\otimes(r-1)}10\psi_1^{\otimes(w-r)}0^{t-2w-1}00^{n-t}} \right) \\ \nonumber 
= \sum_{1 \leq r \leq w, \ 2w+1 \leq t \leq n} & J_{r_1t} \left(\alpha\ket{\psi_1^{\otimes(r-1)}01\psi_1^{\otimes(w-r)}0^{t-2w-1}00^{n-t}} \right. \\ \nonumber
& + \left. \beta\ket{\psi_1^{\otimes(r-1)}00\psi_1^{\otimes(w-r)}0^{t-2w-1}10^{n-t}} \right) \\ \nonumber
& + J_{r_2t} \left(\alpha\ket{\psi_1^{\otimes(r-1)}00\psi_1^{\otimes(w-r)}0^{t-2w-1}10^{n-t}} \right. \\ \nonumber
& + \left. \beta\ket{\psi_1^{\otimes(r-1)}10\psi_1^{\otimes(w-r)}0^{t-2w-1}00^{n-t}} \right) \\ \nonumber
= \sum_{1 \leq r \leq w, \ 2w+1 \leq t \leq n} & J_{r_1t}\alpha\ket{\psi_1^{\otimes(r-1)}01\psi_1^{\otimes(w-r)}0^{t-2w-1}00^{n-t}} \\ \nonumber
& + J_{r_2t}\beta\ket{\psi_1^{\otimes(r-1)}10\psi_1^{\otimes(w-r)}0^{t-2w-1}00^{n-t}} \\ \nonumber
& + (J_{r_1t}\beta + J_{r_2t}\alpha)\ket{\psi_1^{\otimes(r-1)}00\psi_1^{\otimes(w-r)}0^{t-2w-1}10^{n-t}} \\ \nonumber
= \sum_{1 \leq r \leq w} \ket{\psi_1^{\otimes(r-1)}} \otimes & \left(\left(\sum_{2w+1 \leq t \leq n} J_{r_1t}\alpha \right) \ket{01} + \left(\sum_{2w+1 \leq t \leq n} J_{r_2t}\beta \right)\ket{10}\right) \otimes \ket{\psi_1^{\otimes(w-r)}} \otimes \ket{\psi_0}^{\otimes (p-w)} \nonumber \\ 
+ \sum_{1 \leq r \leq w, \ 2w+1 \leq t \leq n} & \left(J_{r_1t}\beta + J_{r_2t}\alpha\right)\ket{\psi_1^{\otimes(r-1)}00\psi_1^{\otimes(w-r)}0^{t-2w-1}10^{n-t}}. \label{constr_5}
\end{align}
Notice that for every $t \in [2w+1, n]$, the state \ket{\psi_1^{\otimes(r-1)}00\psi_1^{\otimes(w-r)}0^{t-2w-1}10^{n-t}} in the above Eq.~(\ref{constr_5}) has the $t^{th}$ qubit set to $1$. However, on the R.H.S. of Eq.~(\ref{eigen_equation}), all the qubits in the positions $2w+1$ through $n$ are set to $0$. Therefore, it implies that
\begin{equation} \label{constr_6}
J_{r_1t}\beta = -J_{r_2t}\alpha,
\end{equation}
for all $r \in [1,w]$ and $t \in [2w+1, n]$. However, notice that if the original $p$-bit string $x$ has the $r^{th}$ bit set to $1$, then Eq.~(\ref{constr_6}) holds true for all $r \in [1,w]$ and $t \neq \{r_1, r_2\}$. In general, for two distinct pairs $u,v \in [1,p]$, following equations hold true:
\begin{align}
& \beta J_{u_1v_1} + \alpha J_{u_2v_1} = 0 \label{pair1} \\
& \beta J_{u_1v_2} + \alpha J_{u_2v_2} = 0 \label{pair2} \\
& \beta J_{v_1u_1} + \alpha J_{v_2u_1} = 0 \label{pair3} \\
& \beta J_{v_1u_2} + \alpha J_{v_2u_2} = 0 \label{pair4}
\end{align}
$\mbox{Eq.~(\ref{pair1})} - \mbox{Eq.~(\ref{pair3})}$ gives
\begin{align}
\alpha J_{u_2v_1} = \alpha J_{v_2u_1},
\end{align}
where we use the symmetricity of the couplings.
Since $|\alpha|^2 + |\beta|^2 = 1$, WLOG we can assume that $\alpha \neq 0$. So, we get
\begin{align}
J_{u_2v_1} = J_{v_2u_1} \label{r_1s_2withr_2s_1}
\end{align}
Then, we can get following relations:
\begin{align}
J_{u_1v_2} = J_{u_2v_1} = -\frac{\beta}{\alpha}J_{u_1v_1} \nonumber \\
J_{u_2v_2} = -\frac{\beta}{\alpha}J_{u_1v_2} = \frac{\beta^2}{\alpha^2}J_{u_1v_1} \label{r1s1andr2s2}
\end{align}
Eq.~(\ref{constr_5}) can now be rewritten as
\begin{align}
\sum_{1 \leq r \leq w} \ket{\psi_1^{\otimes(r-1)}} \otimes & \left(\left(\sum_{2w+1 \leq t \leq n} J_{r_1t}\alpha \right) \ket{01} + \left(\sum_{2w+1 \leq t \leq n} J_{r_2t}\beta \right)\ket{10}\right) \otimes \ket{\psi_1^{\otimes(w-r)}} \otimes \ket{\psi_0}^{\otimes (p-w)}. \label{constr_d}
\end{align}
%From Eq.~(\ref{constr_6}), we can make the following substitutions:
%\begin{align}
%&J_{r_2s_2} = \frac{\beta^2}{\alpha^2}J_{r_1s_1} \quad \forall r,s \in [a,b] \ \text{assuming} \ \alpha \neq 0 \label{r1s1andr2s2} \\
%&J_{r_1s_2} = J_{r_2s_1} \label{r_1s_2withr_2s_1}
%\end{align}
Using Eq.~(\ref{r1s1andr2s2}) in Eq.~(\ref{constr_a}), we get
\begin{align}
= \sum_{1 \leq r < s \leq w} & \bigl(\alpha^2 + \beta^2\bigr)J_{r_1s_1}\ket{\psi_1^{\otimes(r-1)}01\psi_1^{\otimes(s-r-1)}01\psi_1^{\otimes (w-s)}\psi_0^{\otimes (p-w)}}  \nonumber \\
& + 2J_{r_1s_2}\alpha\beta \left(\ket{\psi_1^{\otimes(r-1)}01\psi_1^{\otimes(s-r-1)}10\psi_1^{\otimes (w-s)}\psi_0^{\otimes (p-w)}}\right. \nonumber \\ 
& + \left. \ket{\psi_1^{\otimes(r-1)}10\psi_1^{\otimes(s-r-1)}01\psi_1^{\otimes (w-s)}\psi_0^{\otimes (p-w)}}\right) \nonumber \\ 
& + \left(\alpha^2 + \beta^2 \right)\frac{\beta^2}{\alpha^2}J_{r_1s_1}\ket{\psi_1^{\otimes(r-1)}10\psi_1^{\otimes(s-r-1)}10\psi_1^{\otimes (w-s)}\psi_0^{\otimes (p-w)}}  \label{constr_b}
\end{align}
Further substituting $\gamma_{rs} = \left(\alpha^2 + \beta^2 \right)J_{r_1s_1}$ in Eq.~(\ref{constr_b}), we get
\begin{align}
= \sum_{1 \leq r < s \leq w} & \gamma_{rs}\left(\ket{\psi_1^{\otimes(r-1)}01\psi_1^{\otimes(s-r-1)}01\psi_1^{\otimes (w-s)}\psi_0^{\otimes (p-w)}} \right.  \nonumber \\
& + \left. \frac{\beta^2}{\alpha^2}\ket{\psi_1^{\otimes(r-1)}10\psi_1^{\otimes(s-r-1)}10\psi_1^{\otimes (w-s)}\psi_0^{\otimes (p-w)}} \right) \nonumber \\
& + 2J_{r_1s_2}\alpha\beta \left(\ket{\psi_1^{\otimes(r-1)}01\psi_1^{\otimes(s-r-1)}10\psi_1^{\otimes (w-s)}\psi_0^{\otimes (p-w)}}\right. \nonumber \\ 
& + \left. \ket{\psi_1^{\otimes(r-1)}10\psi_1^{\otimes(s-r-1)}01\psi_1^{\otimes (w-s)}\psi_0^{\otimes (p-w)}}\right) \nonumber \\ 
= \sum_{1 \leq r  \leq w} & \ket{\psi_1^{\otimes(r-1)}}\otimes \ket{01} \otimes \left(\sum_{ s>r} \ket{\psi_1}^{\otimes(s-r-1)} \otimes
\left(\gamma_{rs}\ket{01}+2\alpha\beta J_{r_1s_2} \ket{10}\right) \otimes \ket{\psi_1}^{\otimes (w-s)}\right)\otimes\ket{\psi_0}^{\otimes (p-w)} \nonumber \\ 
& + \ket{\psi_1^{\otimes(r-1)}}\otimes \ket{10} \otimes \left(\sum_{s>r} \ket{\psi_1}^{\otimes(s-r-1)} \otimes
\left(2\alpha\beta J_{r_1s_2}\ket{01}+ \gamma_{rs}\frac{\beta^2}{\alpha^2}\ket{10}\right) \otimes \ket{\psi_1}^{\otimes (w-s)}\right)\otimes\ket{\psi_0}^{\otimes (p-w)} \label{constr_c}
\end{align}

By evaluating the third sum, we get
\begin{align} 
= \sum_{1 \leq r \leq w} & J_{r_1r_2} \SWAP_{r_1r_2} \ket{\psi_1}^{\otimes w} \ket{\psi_0}^{\otimes (p-w)} \nonumber \\ 
= \sum_{1 \leq r \leq w} & J_{r_1r_2} \SWAP_{r_1r_2} \ket{\psi_1}^{\otimes (r-1)} \otimes (\alpha\ket{01}+\beta\ket{10})\otimes
\ket{\psi_1}^{\otimes (w-r)} \otimes \ket{\psi_0}^{\otimes (p-w)} \nonumber \\
= \sum_{1 \leq r \leq w} & J_{r_1r_2} \ket{\psi_1}^{\otimes (r-1)} \otimes (\alpha\ket{10}+\beta\ket{01})\otimes
\ket{\psi_1}^{\otimes (w-r)} \otimes \ket{\psi_0}^{\otimes (p-w)} \label{constr_7}
\end{align} 
Adding Eqs.~(\ref{constr_d}) and (\ref{constr_7}) gives
\begin{align}
\sum_{1 \leq r \leq w} \ket{\psi_1^{\otimes(r-1)}} \otimes & \left(\left(J_{r_1r_2}\beta + \sum_{2w+1 \leq t \leq n} J_{r_1t}\alpha \right) \ket{01}
 + \left(J_{r_1r_2}\alpha + \sum_{2w+1 \leq t \leq n} J_{r_2t}\beta \right)\ket{10}\right) \otimes \ket{\psi_1^{\otimes(w-r)}} \otimes \nonumber \\\ket{\psi_0}^{\otimes (p-w)} \label{constr_e}
\end{align}
Adding Eqs.~(\ref{constr_c}) and (\ref{constr_e}) gives
\begin{align}
\sum_{1 \leq r  \leq w} & \ket{\psi_1^{\otimes(r-1)}}\otimes \ket{01} \otimes \left(\left(\sum_{s>r} \ket{\psi_1}^{\otimes(s-r-1)} \otimes
\left(\gamma_{rs}\ket{01}+2\alpha\beta J_{r_1s_2} \ket{10}\right) \otimes \ket{\psi_1}^{\otimes (w-s)} \right)+ \right. \nonumber \\ 
& \left. \left(J_{r_1r_2}\beta + \sum_{2w+1 \leq t \leq n} J_{r_1t}\alpha \right) \ket{\psi_1^{\otimes(w-r)}} \right)
\otimes\ket{\psi_0}^{\otimes (p-w)} \nonumber \\ 
& + \ket{\psi_1^{\otimes(r-1)}}\otimes \ket{10} \otimes \left(\left(\sum_{s>r} \ket{\psi_1}^{\otimes(s-r-1)} \otimes
\left(2\alpha\beta J_{r_1s_2}\ket{01}+ \gamma_{rs}\frac{\beta^2}{\alpha^2}\ket{10}\right) \otimes \ket{\psi_1}^{\otimes (w-s)} \right)+ \right. \nonumber \\ 
& \left. \left(J_{r_1r_2}\alpha + \sum_{2w+1 \leq t \leq n} J_{r_2t}\beta \right) \ket{\psi_1^{\otimes(w-r)}} \right)
\otimes\ket{\psi_0}^{\otimes (p-w)} \label{constr_g}
\end{align}
Substituting $\rho_{r_1} = \left(J_{r_1r_2}\beta + \sum_{2w+1 \leq t \leq n} J_{r_1t}\alpha \right)$ and $\rho_{r_2} = \left(J_{r_1r_2}\alpha + \sum_{2w+1 \leq t \leq n} J_{r_2t}\beta \right)$ in Eq.~(\ref{constr_g}), we get
\begin{align}
\sum_{1 \leq r  \leq w} & \ket{\psi_1^{\otimes(r-1)}}\otimes \ket{01} \otimes \left(\left(\sum_{s>r} \ket{\psi_1}^{\otimes(s-r-1)} \otimes
\left(\gamma_{rs}\ket{01}+ 2\alpha\beta J_{r_1s_2} \ket{10}\right) \otimes \ket{\psi_1}^{\otimes (w-s)} \right)\right. \nonumber \\ 
& \left. + \rho_{r_1} \ket{\psi_1}^{\otimes (w-r)} \right)\otimes\ket{\psi_0}^{\otimes (p-w)} \nonumber \\ 
& + \ket{\psi_1^{\otimes(r-1)}}\otimes \ket{10} \otimes \left(\left(\sum_{s>r} \ket{\psi_1}^{\otimes(s-r-1)} \otimes
\left(2\alpha\beta J_{r_1s_2}\ket{01}+ \gamma_{rs}\frac{\beta^2}{\alpha^2}\ket{10}\right) \otimes \ket{\psi_1}^{\otimes (w-s)} \right)\right. \nonumber \\ 
&\left. + \rho_{r_2}\ket{\psi_1}^{\otimes (w-r)} \right) \otimes\ket{\psi_0}^{\otimes (p-w)} \label{constr_h}
\end{align}
Evaluating the fourth sum gives
\begin{align}
& \sum_{2w+1 \leq t_1 < t_2 \leq n} \left(J_{t_1t_2} \SWAP_{t_1t_2}\right) \ket{\psi_1}^{\otimes w}\otimes \ket{0}^{\otimes (n-2w)}  \nonumber \\
& = \left( \sum_{2w+1 \leq t_1 < t_2 \leq n} J_{t_1t_2} \right) \ket{\psi_1}^{\otimes w}\otimes \ket{\psi_0}^{\otimes (p-w)} \label{constr_f}
\end{align}
Finally, putting Eqs.~(\ref{constr_h}) and (\ref{constr_f}) together, Eq.~(\ref{eigen_equation}) can be re-written as
\begin{align}
\sum_{1 \leq r  \leq w} & \ket{\psi_1^{\otimes(r-1)}}\otimes \ket{01} \otimes \left(\left(\sum_{s>r} \ket{\psi_1}^{\otimes(s-r-1)} \otimes
\left(\gamma_{rs}\ket{01}+ 2\alpha\beta J_{r_1s_2} \ket{10}\right) \otimes \ket{\psi_1}^{\otimes (w-s)} \right)\right. \nonumber \\ 
& \left. + \rho_{r_1} \ket{\psi_1}^{\otimes (w-r)} \right)\otimes\ket{\psi_0}^{\otimes (p-w)} \nonumber \\ 
& + \ket{\psi_1^{\otimes(r-1)}}\otimes \ket{10} \otimes \left(\left(\sum_{s>r} \ket{\psi_1}^{\otimes(s-r-1)} \otimes
\left(2\alpha\beta J_{r_1s_2}\ket{01}+ \gamma_{rs}\frac{\beta^2}{\alpha^2}\ket{10}\right) \otimes \ket{\psi_1}^{\otimes (w-s)} \right)\right. \nonumber \\ 
&\left. + \rho_{r_2}\ket{\psi_1}^{\otimes (w-r)} \right) \otimes\ket{\psi_0}^{\otimes (p-w)} \nonumber \\
& = \left(\lambda- \sum_{2w+1 \leq t_1 < t_2 \leq n} J_{t_1t_2} \right) \ket{\psi_1}^{\otimes w}\otimes \ket{\psi_0}^{\otimes (p-w)}\label{master}
\end{align}

We perform a change of basis at this point. The new basis states are $\ket{\psi_0} = \ket{00}$, $\ket{\psi_1} = \alpha\ket{01} + \beta\ket{10}$, $\ket{\phi} = \beta^*\ket{01}-\alpha^*\ket{10}$ and $\ket{11}$. Therefore, the old states $\ket{01}$ and $\ket{10}$ can be transformed as follows:
\begin{align}
\ket{01} = \alpha^*\ket{\psi_1} + \beta\ket{\phi} \nonumber \\
\ket{10} = \beta^*\ket{\psi_1} - \alpha\ket{\phi}
\end{align}

Under the basis transformation, Eq.~(\ref{master}) can be written as
\begin{align}
\sum_{1 \leq r  \leq w} & \ket{\psi_1^{\otimes(r-1)}}\otimes (\alpha^*\ket{\psi_1} + \beta\ket{\phi}) \otimes \left(\left(\sum_{s>r} \ket{\psi_1}^{\otimes(s-r-1)} \otimes
\left(\gamma_{rs} (\alpha^*\ket{\psi_1} + \beta\ket{\phi}) \right. \right. \right. \nonumber \\ + \phantom{\sum_{a\le r\le b}}
 &\left. \left. \left. 2\alpha\beta J_{r_1s_2} (\beta^*\ket{\psi_1} - \alpha\ket{\phi})\right) \otimes \ket{\psi_1}^{\otimes (w-s)} \right) + \rho_{r_1} \ket{\psi_1}^{\otimes (w-r)}\right)\otimes\ket{\psi_0}^{\otimes (p-w)} \nonumber \\ 
&+ \ket{\psi_1^{\otimes(r-1)}}\otimes (\beta^*\ket{\psi_1} - \alpha\ket{\phi}) \otimes \left(\left(\sum_{s>r} \ket{\psi_1}^{\otimes(s-r-1)} \otimes
\left(2\alpha\beta J_{r_1s_2}(\alpha^*\ket{\psi_1} + \beta\ket{\phi}) \right. \right. \right. \nonumber \\
& \left. \left. \left. + \gamma_{rs}\frac{\beta^2}{\alpha^2}(\beta^*\ket{\psi_1} - \alpha\ket{\phi})\right) \otimes \ket{\psi_1}^{\otimes (w-s)} \right) + \rho_{r_2}\ket{\psi_1}^{\otimes (w-r)}\right)\otimes\ket{\psi_0}^{\otimes (p-w)} \nonumber \\
& = \left(\lambda- \sum_{2w+1 \leq t_1 < t_2 \leq n} J_{t_1, t_2} \right) \ket{\psi_1}^{\otimes w}\otimes \ket{\psi_0}^{\otimes (p-w)} \label{changed_master}
\end{align}
Notice that R.H.S. of the above Eq.~(\ref{changed_master}) is orthogonal to any state involving $\ket{\phi}$. Therefore, from L.H.S. of Eq.~(\ref{changed_master}), the coefficients associated with the states $\ket{\psi_1}^{\otimes(r-1)} \otimes \ket{\phi} \otimes \ket{\psi_1}^{\otimes(s-r-1)} \otimes \ket{\phi} \otimes \ket{\psi_1}^{\otimes(w-s)} \otimes \ket{\psi_0}^{\otimes(p-w)}$ for each distinct $(r,s)$ pair can be equated to zero. So, we get
%\begin{align}
%\sum_{a \leq r  \leq b} & \ket{\psi_1}^{\otimes(r-1)}\otimes \ket{\phi} \otimes \left(\sum_{s>r} \ket{\psi_1}^{\otimes(s-r-1)} \otimes
%\left(\gamma_{rs}\beta^2 -2\alpha^2\beta^2 J_{r_1s_2}\right) \ket{\phi}\otimes \ket{\psi_1}^{\otimes (w-s)} \right) \otimes \ket{\psi_0}^{\otimes (p-w)} \nonumber \\
%& + \ket{\psi_1}^{\otimes(r-1)}\otimes \ket{\phi} \otimes \left(\sum_{s>r} \ket{\psi_1}^{\otimes(s-r-1)} \otimes
%\left(-2\alpha^2\beta^2 J_{r_1s_2} + \gamma_{rs}\beta^2\right) \ket{\phi}\otimes \ket{\psi_1}^{\otimes (w-s)} \right) \otimes \ket{\psi_0}^{\otimes (p-w)} = 0 \nonumber \\
%\sum_{a \leq r  \leq b} & \ket{\psi_1}^{\otimes(r-1)}\otimes \ket{\phi} \otimes \left(\sum_{s>r} \ket{\psi_1}^{\otimes(s-r-1)} \otimes
%\left(2\gamma_{rs}\beta^2 -4\alpha^2\beta^2 J_{r_1s_2} \right) \ket{\phi}\otimes \ket{\psi_1}^{\otimes (w-s)} \right) \otimes \ket{\psi_0}^{\otimes (p-w)} =0 \nonumber 
%\end{align}
\begin{align}
&\beta^2\gamma_{rs} - 2\alpha^2 \beta^2 J_{r_1s_2} - 2\alpha^2 \beta^2 J_{r_1s_2} + \gamma_{rs} \beta^2 = 0 \nonumber \\
&\beta^2\gamma_{rs} = 2\alpha^2\beta^2J_{r_1s_2} \nonumber \\ 
& \beta^2(\alpha^2 + \beta^2)J_{r_1s_1} = 2\alpha^2\beta^2J_{r_1s_2} \label{constr_i}
\end{align}
%Second, equating the coefficient of the state $\ket{\psi_1}^{\otimes(r-1)} \otimes \ket{\phi} \otimes \ket{\psi_1}^{\otimes(w-r)} \otimes \ket{\psi_0}^{\otimes(p-w)}$ to zero, we get
%\begin{align}
%& \beta \rho_{r_1} - \alpha \rho_{r_2} = 0 \nonumber \\
%& \beta \rho_{r_1} = \alpha \rho_{r_2} \quad \forall a \leq r \leq b \label{rho_eqn}
%\end{align}
Using the relations $J_{r_2s_2} = \frac{\beta^2}{\alpha^2}J_{r_1s_1}$, $J_{r_1s_2} = J_{r_2s_1}$ in Eq.~(\ref{constr_4}), we have
\begin{align}
& \beta J_{r_1s_1}\left(\frac{\alpha^2 + \beta^2}{\alpha}\right) = -\left(\alpha^2 + \beta^2\right) J_{r_1s_2} \label{alpha_beta_ratio}
\end{align}
% & \frac{J_{r_1s_1}}{J_{r_1s_2}} = -\frac{\alpha}{\beta} 
Multiplying Eq.~(\ref{alpha_beta_ratio}) by $\beta$ on both sides, we get 
\begin{equation}
\beta^2 J_{r_1s_1} \left(\frac{\alpha^2 + \beta^2}{\alpha}\right) = -\beta\left(\alpha^2 + \beta^2\right)J_{r_1s_2} \label{intermediate_eqn}
\end{equation}

Substituting the value of $\beta^2\left(\alpha^2 + \beta^2\right)J_{r_1s_1}$ from Eq.~(\ref{constr_i}) in Eq.~(\ref{intermediate_eqn}), we get
\begin{align}
& 2\alpha \beta^2 J_{r_1s_2} = -\beta\left(\alpha^2 + \beta^2\right)J_{r_1s_2} \nonumber \\
& \beta J_{r_1s_2} \left(\alpha + \beta \right)^2 = 0 \nonumber 
\end{align}
If we assume that $J_{r_1s_2} \neq 0$, we conclude that either $\beta = 0$ or $\alpha = -\beta$.
If $\beta = 0$, from Eq.~(\ref{constr_6}), $J_{r_2t} = 0$ for all $r \in [1,w]$ and $t \in [1, n] - \{ r_1, r_2 \}$. It then follows that $J_{r_2s_1} = J_{r_1s_2} = 0$, which contradicts our assumption that $J_{r_1s_2} \neq 0$. Therefore, $\beta$ cannot be zero. 
At this point, we assume that the couplings are non zero and we conclude that $\alpha = -\beta$.
Now we aim to deduce constraints on the eigenvalue $\lambda$ by equating the coefficients of the states that don't involve the state $\ket{\phi}$ from Eq.~(\ref{changed_master}).
\begin{align}
& \sum_{1 \leq r \leq w} \alpha^* \ket{\psi_1}^{\otimes r} \otimes\left(\rho_{r_1} + \sum_{s > r} \left(\gamma_{rs}\alpha^* + 2\alpha|\beta|^2 J_{r_1s_2} \right)\right) \ket{\psi_1}^{\otimes {w-r}}\otimes \ket{\psi_0}^{\otimes (p-w)} \nonumber \\
& \quad \quad + \beta^* \ket{\psi_1}^{\otimes r} \otimes \left(\rho_{r_2} + \sum_{s > r} \left(2|\alpha|^2 \beta J_{r_1s_2} + \gamma_{rs}|\beta|^2 \frac{\beta}{\alpha^2}\right)\right) \ket{\psi_1}^{\otimes {w-r}} \otimes \ket{\psi_0}^{\otimes (p-w)} \nonumber \\ 
& = \left(\lambda- \sum_{2w+1 \leq t_1 < t_2 \leq n} J_{t_1, t_2} \right) \ket{\psi_1}^{\otimes w} \otimes \ket{\psi_0}^{\otimes (p-w)} \nonumber \\
& \left(\sum_{1 \leq r \leq w} \rho_{r_1}\alpha^* + \rho_{r_2}\beta^*  + \sum_{1 \leq r \leq w} \sum_{s > r} 4 |\alpha|^2 |\beta|^2 J_{r_1s_2} + \left(1 + \frac{\beta^2}{\alpha^2}\right)\left(|\alpha|^4 + |\beta|^4\right)J_{r_1s_1} \right)\ket{\psi_1}^{\otimes w} \otimes \ket{\psi_0}^{\otimes (p-w)} \nonumber \\
& = \left(\lambda- \sum_{2w+1 \leq t_1 < t_2 \leq n} J_{t_1, t_2} \right) \ket{\psi_1}^{\otimes w} \otimes \ket{\psi_0}^{\otimes (p-w)} \label{lambda_equation}
\end{align}

%We have two cases, and we evaluate both cases one by one.
%Case when $\beta = 0$:
% So, it follows from Eq.~(\ref{lambda_equation}) that
%\begin{align}
%& \left(\sum_{a \leq r \leq b} \quad \sum_{2w+1 \leq t \leq n} J_{r_1t}\right) + \left(\sum_{a \leq r \leq b} \quad \sum_{s > r} J_{r_1s_1}\right)  = \left(\lambda_w- \sum_{2w+1 \leq t_1 < t_2 \leq n} J_{t_1, t_2} \right) \nonumber \\
%& \lambda_w = \left(\sum_{a \leq r \leq b} \quad \sum_{2w+1 \leq t \leq n} J_{r_1t}\right) + \left(\sum_{a \leq r \leq b} \quad \sum_{s > r} J_{r_1s_1}\right) + \left(\sum_{2w+1 \leq t_1 < t_2 \leq n} J_{t_1, t_2} \right) \label{eigenvalue_eqn}
%\end{align}
%From Eq.~(\ref{rho_eqn}), we get
%\begin{align}
%&\alpha \rho_{r_2} = 0 \nonumber \\
%&\rho_{r_2} = 0 \quad (\text{since} \ \alpha \neq 0 ) \nonumber \\
%& J_{r_1r_2} \alpha = 0 \nonumber \\
%& J_{r_1r_2} = 0 \quad \forall a \leq r \leq b \quad (\text{since} \ \alpha \neq 0) \nonumber
%\end{align}
%Therefore, Eq.~(\ref{eigenvalue_eqn}) suggests that the eigenvalue $\lambda_w$ is equal to the sum of non zero couplings.
%\begin{align}
%\lambda_w = \sum_{1 \leq q_1 < q_2 \leq n} J_{q_1q_2} - 2\sum_{a \leq r \leq b} \sum_{s>r}J_{r_1s_2}
%\end{align}

When $\alpha = - \beta$ and $|\alpha|^2 + |\beta|^2 = 1$, $|\alpha|^2 = |\beta|^2 = \frac{1}{2}$. Then, from Eq.~(\ref{lambda_equation}), we get
\begin{align}
& \lambda = \left(\sum_{1 \leq r \leq w} \left(\rho_{r_1} - \rho_{r_2}\right)\alpha^* \right) +  \left(\sum_{1 \leq r \leq w} \sum_{s > r} J_{r_1s_1} + J_{r_1s_2} \right) + \left( \sum_{2w+1 \leq t_1 < t_2 \leq n} J_{t_1t_2} \right) \label{almost_final_eigen_eqn}
\end{align}
Notice that $\rho_{r_1} = \left(J_{r_1r_2}\beta + \sum_{2w+1 \leq t \leq n} J_{r_1t}\alpha \right)$ and $\rho_{r_2} = \left(J_{r_1r_2}\alpha + \sum_{2w+1 \leq t \leq n} J_{r_2t}\beta \right)$. So, when $\alpha = -\beta$, 
\begin{align}
\left(\rho_{r_1} - \rho_{r_2}\right)\alpha^* & = -2J_{r_1r_2}|\alpha|^2 + \sum_{2w+1 \leq t \leq n} (J_{r_1t} + J_{r_2t})|\alpha|^2 \nonumber \\
& = -J_{r_1r_2} + \frac{1}{2}\sum_{2w+1 \leq t \leq n} \left(J_{r_1t} + J_{r_2t}\right) \quad (\text{since} \ |\alpha|^2 = \frac{1}{2})
\end{align}
When $\alpha = - \beta$, following relations hold true among the couplings:
From Eq.~(\ref{r1s1andr2s2}), we get 
\begin{align}
J_{u_1v_1} = J_{u_2v_2}, \label{parallel_couplings}
\end{align}
and from Eqs.~\ref{pair1} and \ref{pair3}, we get
\begin{align}
& J_{u_1v_1} = J_{u_1v_2} \ \text{and} \label{r_1_eqn} \\
& J_{u_1v_1} = J_{u_2v_1} \label{s_1_eqn} \quad \text{for distinct pairs} \ u,v \in [1,p].
\end{align}
Putting together Eqs.~\ref{parallel_couplings}, \ref{r_1_eqn}, and \ref{s_1_eqn}, it shows that for any distinct pairs $u,v \in [1,p]$, all the external couplings are equal, i.e. $J_{u_1v_1} = J_{u_1v_2} = J_{u_2v_1} = J_{u_2v_2}$.
Therefore, Eq.~(\ref{almost_final_eigen_eqn}) can be simplified to
\begin{align}
& \lambda = \left(\sum_{1 \leq r \leq w} -J_{r_1r_2} + \sum_{2w+1 \leq t \leq n} J_{r_1t} \right) + 2\left(\sum_{1 \leq r \leq w} \sum_{s > r} J_{r_1s_1} \right) + \left( \sum_{2w+1 \leq t_1 < t_2 \leq n} J_{t_1t_2} \right) \nonumber \\
& \lambda = \left(\sum_{1 \leq r \leq w} \ \sum_{2w+1 \leq t \leq n} J_{r_1t} \right) + 2\left(\sum_{1 \leq r \leq w} \sum_{s > r} J_{r_1s_1} \right) + \left( \sum_{2w+1 \leq t_1 < t_2 \leq n} J_{t_1t_2} \right) - \left(\sum_{1 \leq r \leq w} J_{r_1r_2}\right) \nonumber \\
& \lambda = \left(\sum_{1 \leq r \leq w} \ \sum_{w+1 \leq t \leq p} J_{r_1t_1} + J_{r_1t_2} \right) + 2\left(\sum_{1 \leq r \leq w} \sum_{s > r} J_{r_1s_1} \right) + \left( \sum_{w+1 \leq m<n \leq p} J_{m_1n_1}+J_{m_1n_2}+J_{m_2n_1}+J_{m_2n_2}\right) \nonumber \\
& \quad \quad \quad \quad + \left(\sum_{w+1 \leq m \leq p} J_{m_1m_2} \right) - \left(\sum_{1 \leq r \leq w} J_{r_1r_2}\right) \nonumber \\
\lambda & = 2\left(\sum_{1 \leq r \leq w} \ \sum_{w+1 \leq t \leq p} J_{r_1t_1} \right) + 2\left(\sum_{1 \leq r \leq 1} \sum_{s > r} J_{r_1s_1} \right) + 4\left(\sum_{w+1 \leq m<n \leq p} J_{m_1n_1}\right)
+ \left(\sum_{w+1 \leq m \leq p} J_{m_1m_2} \right) \nonumber \\
& \qquad - \left(\sum_{1 \leq r \leq w} J_{r_1r_2}\right) \label{final_eigenvalue_eqn}
\end{align}

\section{Implications of couplings only depending on the Hamming weight} \label{Hamming_weight_dependence}
\begin{Thm} \label{Thm_Hamming_weight}
For any $p$-bit ($p \geq 2$) string $x$ encoded into the logical state $\ket{x_L}$ (given by Eq.~(\ref{encoded_state})), if $\ket{x_L}$ is an eigenstate of $J^2$ (given by Eq.~(\ref{modified_J2})) such that the eigenvalue $\lambda_x$ depends only on the Hamming weight of the string $x$, then the internal (respectively external) couplings of any distinct qubit pairs $u,v \in [1,p]$ are equal to the internal (respectively external) couplings of any other distinct qubit pairs $u^\prime,v^\prime \in [1,p]$.
\end{Thm}
\begin{proof}
Consider two strings $x$ and $y$ of the same Hamming weight $w$ such that the Hamming distance between them is $2$. Assume that the $x_u= 1$ and $x_v = 0$ in $x$, and $y_u= 0$ and $y_v = 1$ in $y$. Let $C_1^x$ and $C_0^x$ denote the sets of pair indices which are set to states $\ket{\psi_1}$ and $\ket{00}$ respectively for the string $x$, and similarly define $C_1^y$ and $C_0^y$ for the string $y$. Notice that the sets $C_1^x - \{u\}$ and $C_0^x - \{v\}$ are identical with the sets $C_1^y - \{v\}$ and $C_0^y - \{u\}$ respectively. We now write Eq.~(\ref{final_eigenvalue_eqn}) for the string $x$ as below:
\begin{align}
\lambda_w^x & = 2\left(\sum_{\substack{r \in C_1^x \\ r \neq u }} \ \sum_{\substack{ t \in C_0^x \\ t \neq v}} J_{r_1t_1} \right) +
2\left(\sum_{\substack{ t \in C_0^x \\ t \neq v}} J_{u_1t_1} \right) + 
2\left(\sum_{\substack{r \in C_1^x \\ r \neq u }} J_{r_1v_1} \right) + 
2J_{u_1v_1} + 2\left(\sum_{\substack{r,s \in C_1^x \\ r < s}} J_{r_1s_1} \right)\nonumber \\
& + 4\left(\sum_{\substack{m, n \in C_0^x \\ m<n}} J_{m_1n_1}\right) +  \left(\sum_{\substack{m \in C_0^x}} J_{m_1m_2} \right) 
- \left(\sum_{r \in C_1^x} J_{r_1r_2}\right) \label{x_eqn}
\end{align}
Similarly, for the string $y$, we get
\begin{align}
\lambda_w^y & = 2\left(\sum_{\substack{r \in C_1^y \\ r \neq v }} \ \sum_{\substack{ t \in C_0^y \\ t \neq u}} J_{r_1t_1} \right) +
2\left(\sum_{\substack{ t \in C_0^y \\ t \neq u}} J_{v_1t_1} \right) + 
2\left(\sum_{\substack{r \in C_1^y \\ r \neq v }} J_{r_1u_1} \right) + 
2J_{v_1u_1} + 2\left(\sum_{\substack{r,s \in C_1^y \\ r < s}} J_{r_1s_1} \right)\nonumber \\
& + 4\left(\sum_{\substack{m, n \in C_0^y \\ m<n}} J_{m_1n_1}\right) +  \left(\sum_{\substack{m \in C_0^y}} J_{m_1m_2} \right) 
- \left(\sum_{r \in C_1^y} J_{r_1r_2}\right) \label{y_eqn}
\end{align}
Subtracting Eq.~(\ref{y_eqn}) from Eq.~(\ref{x_eqn}), we get
\begin{align}
\lambda_w^x - \lambda_w^y & = 2\left(\sum_{\substack{ t \in C_0^x \\ t \neq v}} J_{u_1t_1} \right) - 2\left(\sum_{\substack{ t \in C_0^y \\ t \neq u}} J_{v_1t_1} \right) +
2\left(\sum_{\substack{r \in C_1^x \\ r \neq u }} J_{r_1v_1} \right) -
2\left(\sum_{\substack{r \in C_1^y \\ r \neq v }} J_{r_1u_1} \right) \nonumber \\
& + 2\left(\sum_{\substack{s \in C_1^x \\ s \neq u}} J_{u_1s_1} \right) - 2\left(\sum_{\substack{s \in C_1^y \\ s \neq v}} J_{v_1s_1} \right)
+ 4\left(\sum_{\substack{n \in C_0^x \\ n \neq v}} J_{v_1n_1}\right) - 4\left(\sum_{\substack{n \in C_0^y \\ n \neq u}} J_{u_1n_1}\right)
+ 2(J_{v_1v_2} - J_{u_1u_2}) \nonumber \\
& = 2\left(\sum_{\substack{ z \in C_0^x \cup C_1^x \\ z \neq u}} J_{u_1z_1} \right) - 2\left(\sum_{\substack{ z \in C_0^y \cup C_1^y \\ z \neq v}} J_{v_1z_1} \right) + 2\left(\sum_{\substack{z \in C_0^x \cup C_1^x \\ z \neq v}} J_{v_1z_1}\right) - 2\left(\sum_{\substack{z \in C_0^y \cup C_1^y \\ z \neq u}} J_{u_1z_1} \right) \nonumber \\
& + 2\left(\sum_{\substack{n \in C_0^x \\ n \neq v}} J_{v_1n_1}\right) - 2\left(\sum_{\substack{n \in C_0^y \\ n \neq u}} J_{u_1n_1}\right) + 2(J_{v_1v_2} - J_{u_1u_2}) \nonumber \\
& = 2\left(\sum_{\substack{n \in C_0^x \\ n \neq v}} J_{v_1n_1}\right) - 2\left(\sum_{\substack{n \in C_0^y \\ n \neq u}} J_{u_1n_1}\right) +  2(J_{v_1v_2} - J_{u_1u_2}) \label{uv_eqn}
\end{align}

If we consider the strings $x^\prime$ and $y^\prime$ both of Hamming weight $1$ such that $C_1^{x^\prime} = \{ u \}$ and $C_1^{y^\prime} = \{ v \}$, then from Eq.~(\ref{uv_eqn}), we get
\begin{align}
\lambda_1^{x^\prime} - \lambda_1^{y^\prime} = 2\left(\sum_{\substack{n \neq \{u,v\}}} J_{v_1n_1}\right) - 2\left(\sum_{\substack{n \neq \{u,v\}}} J_{u_1n_1}\right) +  2(J_{v_1v_2} - J_{u_1u_2}) \label{x_prime_y_prime_eqn}
\end{align}
If we consider the strings $x^{\prime\prime}$ and $y^{\prime\prime}$ both of Hamming weight $2$ such that $C_1^{x^{\prime\prime}} = 
\{ u, t \}$ and $C_1^{y^{\prime\prime}} = 
\{ v, t \}$, then from Eq.~(\ref{uv_eqn}), we get
\begin{align}
\lambda_2^{x^{\prime\prime}} - \lambda_2^{y^{\prime\prime}} = 2\left(\sum_{\substack{n \neq \{u,v,t\}}} J_{v_1n_1}\right) - 2\left(\sum_{\substack{n \neq \{u,v,t\}}} J_{u_1n_1}\right) +  2(J_{v_1v_2} - J_{u_1u_2}) \label{xy_double_prime}
\end{align}
If we assume that the eigenvalues only depend on the Hamming weight of the input string, then from Eqs.~\ref{x_prime_y_prime_eqn} and \ref{xy_double_prime}, we get
\begin{align}
& \left(\sum_{\substack{n \neq \{u,v\}}} J_{v_1n_1}\right) - \left(\sum_{\substack{n \neq \{u,v\}}} J_{u_1n_1}\right) = \left(\sum_{\substack{n \neq \{u,v,t\}}} J_{v_1n_1}\right) - \left(\sum_{\substack{n \neq \{u,v,t\}}} J_{u_1n_1}\right) \nonumber \\
& \quad J_{v_1t_1} = J_{u_1t_1}.  \label{couplings_equality}
\end{align}
Since the choice of the pair indices $u,v$ and $t$ was completely arbitrary, we conclude that Eq.~(\ref{couplings_equality}) holds true for any three distinct pairs, which then implies that all the external couplings are indeed equal. As a result, we conclude from Eq.~(\ref{uv_eqn}) that $J_{v_1v_2} = J_{u_1u_2}$for any two arbitrary distinct pairs $u$ and $v$. Therefore, all the internal couplings are also equal.
\end{proof}

\end{document}